\documentclass[11pt,pdftex,a4paper]{article}
\usepackage{verbatim}
\usepackage{amssymb}
\usepackage{multicol}
\usepackage{amsmath}
\usepackage{tikz}
\usepackage{color} 
\usepackage{fullpage}
\newif\ifcode
\codefalse
\codetrue
\usepackage{macros}
\ifcode
\usepackage[ruled]{algorithm}
\usepackage{algpseudocode}
\usepackage{ioa_code}

\newtheorem{theorem}{Theorem}
\newtheorem{definition}[theorem]{Definition}
\newtheorem{lemma}[theorem]{Lemma}

\newenvironment{proof}[1][Proof]{\noindent\textbf{#1.} }{\hfill $\Box$\\[2mm]} %\rule{0.5em}{0.5em}\\}
\newenvironment{proofsketch}[1][Proof sketch]{\noindent\textbf{#1.} }{\hfill $\Box$\\[2mm]} %\rule{0.5em}{0.5em}\\}
\newenvironment{reptheorem}[1][Theorem]{\noindent\textbf{#1}}{}

\newcounter{linenumber}

\def\Nat{\ensuremath{\mathbb{N}}}

\newcommand{\true}{\mathit{true}}
\newcommand{\false}{\mathit{false}}

\newcommand{\remove}[1]{}

\newcommand{\Wset}{\textit{Wset}}
\newcommand{\Rset}{\textit{Rset}}
\newcommand{\Dset}{\textit{Dset}}

\newcommand{\parts}{\textit{parts}}

\newcommand{\Read}{\textit{read}}
\newcommand{\Write}{\textit{write}}
\newcommand{\TryC}{\textit{tryC}}
\newcommand{\TryA}{\textit{tryA}}
\newcommand{\ok}{\textit{ok}}

\newcommand{\multitrylock}{\textit{multi-trylock}}
\newcommand{\CAS}{\textit{CAS}}
\newcommand{\mCAS}{\textit{mCAS}}

\newcommand{\ignore}[1]{}

\begin{document}

\bibliographystyle{plain}

\title{On the Cost of Concurrency in Transactional Memory} 

\author{Petr Kuznetsov\\ 
\\
\normalsize TU Berlin/Deutsche Telekom Laboratories
\and 
Srivatsan Ravi\\
\\
\normalsize TU Berlin/Deutsche Telekom Laboratories
%\thanks{TU Berlin/ Deutsche Telekom Laboratories
%FG INET, Sekr. TEL 16, Ernst-Reuter-Platz 7, 10587 Berlin; Email: srivatsan.ravi@net.t-labs.tu-berlin.de}
\thanks{The research leading to these results has received funding
  from the European Union Seventh Framework 
Programme (FP7/2007-2013) under grant agreement N 238639, ITN project TRANSFORM.}
}

%\institute{Deutsche Telekom Laboratories/TU Berlin}

%%%%%%%%%%%%%%%%%%%%%%%%%%%%%%%%%%%%%%%%%%%%%%%%%%%%%%%%%%%%%%%%%%%%%%%%%%%%%%%%
\date{}
\maketitle

\begin{abstract}
The promise of software transactional memory (STM) is 
to combine an easy-to-use programming interface
with an efficient utilization of the concurrent-computing abilities 
provided by modern machines.
But does this combination come with an inherent cost?

We evaluate the cost of concurrency by measuring the amount of expensive synchronization
that must be employed in an STM implementation that ensures \emph{positive} concurrency, 
i.e., allows for concurrent transaction processing in some executions. 
We focus on two popular progress conditions that provide positive concurrency: 
\emph{progressiveness} and \emph{permissiveness}.

We show that in permissive STMs, providing a very high degree of concurrency, 
a transaction performs a linear number of expensive 
synchronization patterns with respect to its read-set size. 
In contrast, progressive STMs provide a very small degree of concurrency
but, as we demonstrate, can be implemented using at most one expensive synchronization 
pattern per transaction. 
However, we show that even in progressive STMs, a transaction has 
to ``protect'' (e.g., by using locks or strong synchronization primitives)
a linear amount of data with respect to its write-set size. 
Our results suggest that looking for high degrees of concurrency
in STM implementations may bring a considerable synchronization cost.
%Our results suggest that maximizing the ability to process multiple transactions in parallel would involve employing a significant number of
%expensive synchronization patterns per transaction.
\end{abstract}

\begin{center}
\end{center}
%
%\begin{center}
%Regular paper, eligible for the best student paper award (Srivatsan Ravi is a full-time student).
%\end{center}
%
\thispagestyle{empty}
\clearpage
\pagenumbering{arabic}
%%%%%%%%%%%%%%%%%%%%%%%%%%%%%%%%%%%%
\section{Introduction}
\label{sec:intro}
%%%%%%%%%%%%%%%%%%%%%%%%%%%%%%%%%%%%
The software transactional memory (STM) paradigm promises to efficiently 
exploit the concurrency provided  
by modern computers while offering an easy-to-use programming interface.
It allows a programmer to write a concurrent program as a sequence of 
\emph{transactions}.
A transaction is a series of read and write operations on
\emph{transactional objects} (or \emph{t-objects}).
An STM implementation turns this series into
a sequence of accesses to underlying \emph{base objects} and exports  
``all-or-nothing'' semantics: 
%that can be seen as executed atomically.
every transaction either \emph{commits} in which case all its operations 
are expected to instantaneously ``take effect'', or \emph{aborts} in which case the transaction 
does not affect any other transaction.
In this paper, the default STM correctness property is \emph{opacity}~\cite{GK08-opacity,tm-book} 
that, informally, requires that in every execution, there is a total order on 
all transactions, \emph{including aborted ones}, 
%which respects the real-time order of  non-overlapping transactions 
where every read operation returns the argument of 
the last committed write operation on the read t-object.
%A correctness property of STM specifies what a read 
%operation may return in any given run. 
%A popular correctness condition is \emph{opacity}~\cite{GK08-opacity} that, informally, 
%requires that in every execution, there is a total order on 
%all transactions, including aborted ones, which respects 
%the real-time order of  non-overlapping transactions 
%and every read operation returns the value resulted after the last committed write operation. 

An STM implementation that aborts every transaction is trivially correct but useless. 
Therefore, we need to specify a \emph{progress condition} that captures the execution scenarios in
which a transaction should commit.
%A trivial progress condition would allow a transaction to abort in any execution. 
Consider, for example, a simple non-trivial progress condition that 
requires a transaction to commit if it does not overlap with any other 
transaction.
This condition can be implemented using a single lock that is acquired at 
the beginning of a transaction and released at its end. 
The resulting ``single-lock'' STM will be running one transaction at a time, thus ignoring 
the potential benefits of multiprocessing.
Similarly, an \emph{obstruction-free} STM~\cite{OFTM}
that only requires a transaction to commit if it 
eventually runs with no contention allows for no concurrency at all. 
But to exploit the power of modern multiprocessor machines, 
an STM implementation must allow at least \emph{some} transactions to make progress
concurrently. 
If this is the case, we say that the implementation provides \emph{positive}
concurrency, in contrast to zero concurrency provided by 
``single-lock'' and obstruction-free STMs. 

In this paper, we try to understand the inherent costs of 
allowing multiple concurrent transactions to commit.  
Therefore, we focus on progress conditions that provide 
positive concurrency: 
\emph{progressiveness}~\cite{GK09-progressiveness} 
and \emph{permissiveness}~\cite{GHS08-permissiveness}.
Informally, a progressive STM~\cite{GK09-progressiveness}
provides a very small degree of concurrency by only enforcing a transaction $T$
to commit if it encounters no concurrent \emph{conflicting} transaction $T'$: 
$T$ and $T'$ conflict on a t-object $X$ if they concurrently access $X$ 
and one of the transactions tries to update $X$. 
A stronger variant of progressiveness, called \emph{strong progressiveness}, 
additionally requires that in case a set of transactions 
conflict on at most one t-object, at least one transaction commits.
A much more demanding permissive STM~\cite{GHS08-permissiveness}
stipulates that a transaction must commit, unless 
committing it violates correctness, which, informally, provides
the highest degree of concurrency.

%We consider STM implementations in shared-memory systems: 
%operations on STM are implemented as sequences of accesses to
%\emph{base} shared objects.
To understand the inherent cost of positive concurrency
in STM implementations, %we should select appropriate metrics.
we first consider the number of \emph{RAW/AWAR} 
synchronization patterns~\cite{AGK11-popl} 
that must be performed by a process in the course of a transaction.
A \emph{read-after-write} (RAW) pattern consists of a write
to a (shared) base object $x$ followed 
by a read from a different base object $y$ (without a write to $y$ in between). 
An \emph{atomic write-after-read} (AWAR) pattern consists of 
an atomic (indivisible) execution of a read 
of a base object followed by a write on (possibly the same) base object.  
Accounting for RAW/AWAR patterns is important since most modern processor 
architectures use relaxed memory models, where maintaining the order of operations 
in a RAW requires a \emph{memory fence}~\cite{McKenney10}  
and each AWAR is manifested as an atomic instruction such as 
Compare-and-Swap (CAS). 
In most architectures, memory fences and atomic instructions 
are believed to be considerably slower than
regular shared-memory accesses~\cite{AdveG96,Lee99,McKenney10,LWN-mckenney}.

We show that every permissive 
and opaque STM implementation 
has, for any $m\in\Nat$, an execution in which a transaction with a read set of size $m$ 
incurs $\Omega(m)$ consecutive RAW/AWAR patterns.
%, where $\textit{Vars}(T)$ denotes the \emph{data set} of $T$, i.e., 
%the set of \emph{transactional} objects that $T$ reads or updates. 
This contrasts with a single-lock STM that uses only one
such pattern, since a successful lock acquisition can be implemented using 
only one (multi-) RAW~\cite{Lamport74-bakery}\footnote{A multi-RAW
consists of a series of writes followed by a series of reads from a distinct locations.
Maintaining the multi-RAW order can be achieved with a single memory fence.}
or AWAR~\cite{anderson-90-tpds}.
We show that one RAW/AWAR is in fact optimal for single lock STMs.
%In fact, we describe a progressive STM that uses at most one RAW per transaction, regardless 
%of the transaction's data set.
Moreover, we present implementations of \emph{progressive} STMs 
that employ just a single RAW or AWAR pattern per transaction.
% Strong progressive implementation assumes starvation-freedom, but no need to mention in Introduction since paper assumes starvation-freedom
Also, we describe a \emph{strongly progressive} space-bounded STM implementation that incurs four RAWs per transaction. 
%The first implementation assumes that an atomic $m$-CAS primitive~\cite{AHCAS}
%that can access an arbitrary set of $m$ base objects is available, 
%and incurs just a single AWAR (contained in $m$-CAS) per transaction
%we show that a progressive STM can also be implemented using 
%just a single AWAR (contained in $m$-CAS) per transaction,
%regardless of the transaction's data-set size.
%Similarly, we describe a read-write implementation of a progressive STM 
%that employs at most one multi-RAW per transaction.    

These implementations suggest that the RAW/AWAR metric is too coarse-grained 
to evaluate the complexity of progressive STMs.
Therefore, we introduce a new metric called \emph{protected data size}
that, intuitively, captures the amount of data that a transaction 
must exclusively control at some point of its execution. 
%Informally, a transaction $T$ \emph{protects} a transactional object $X$ 
%in an execution fragment $\pi$ if $T$ atomically changes the value of $X$ 
%in the next step or it does not allow any concurrent transaction to read $X$. 
%%E.g., the $k$-CAS atomic primitive used in a progressive STM implementation
%%mentioned above   protects a set of $k$ objects: 
%%no transaction can access an object in the set between the  
%%Similarly, 
All progressive STM implementations we are aware of 
(see, e.g., an overview in~\cite{GK09-progressiveness})
use locks or timing assumptions to give an updating transaction exclusive access  
to all objects in its write set at some point of its execution.
E.g., lock-based progressive implementations require that a
transaction grabs all locks on its write set
before updating  the corresponding base objects.    
Our results show that this is an inherent price to pay for
providing progressive concurrency: every committed transaction in a progressive and 
\emph{strict disjoint-access-parallel}\footnote{A disjoint-access-parallel
STM implementation~\cite{israeli-disjoint,AHM09} guarantees that
transactions accessing 
disjoint sets of transactional objects are executed independently of each other, i.e., 
without conflicting on the base objects.} STM implementation 
must, at some point of its execution, protect every object 
in its write set.
Interestingly, as our progressive implementations show, 
the transaction's read set does not need to be protected.

In brief, our results imply that providing high degrees of concurrency 
in opaque STM implementations incurs a considerable synchronization cost.
Permissive STMs, while providing the best possible concurrency in theory,
require a strong synchronization primitive or a memory fence per
read operation, which may result in excessively slow execution times.
Progressive STMs provide only basic concurrency but perform considerably better in this respect: 
we present progressive implementations that incur constant RAW/AWAR complexity.
Does this mean that maximizing the ability of processing multiple transactions
in parallel should not be an important factor in STM design?
Should we rather assume little positive concurrency provided by progressiveness or  
even focus on speculative single-lock solutions \'a la \emph{flat combining}~\cite{HendlerIST10}?    
Difficult to say affirmatively, but our results suggest so.  

The rest of the paper is organized as follows. Section~\ref{sec:model} 
briefly introduces our system model and recalls the correctness criteria
in STM. Section~\ref{sec:prel} presents some useful properties of STM implementations and Section~\ref{sec:liveness} recalls the definitions of progress 
conditions of STM, including progressiveness and permissiveness.
Section~\ref{sec:raw} presents the definitions of RAW/AWAR complexity.
Sections~\ref{sec:perm} presents a linear lower bound on the number of RAW/AWAR patterns 
executed by a transaction in a permissive STM. 
Section~\ref{sec:prog} describes our progressive STM implementations 
that perform constant RAWs or AWARs per transaction and 
presents a lower bound on the amount of data to be protected by a transaction in
a progressive STM.
Section~\ref{sec:related}
summarizes some related work and Section~\ref{sec:conclusion} concludes the paper. 
Detailed proofs are delegated to the optional Appendix.
%The full version of this paper, with some extensions and side results, 
%can be found in~\cite{KR11}.
%The optional Appendix contains the Lemmas on some STM properties omitted from the main paper and the
%complete proofs of correctness for our algorithms. 
%%%%%%%%%%%%%%%%%%%%%%%%%%%%%%%%%%%%
%\section{Model}
%\label{sec:model}
%%%%%%%%%%%%%%%%%%%%%%%%%%%%%%%%%%%%
%%%%%%%%%%%%%%%%%%%%%%%%%%%%%%%%%%%%
\section{Model}
\label{sec:model}
%%%%%%%%%%%%%%%%%%%%%%%%%%%%%%%%%%%%
%In this section we introduce the basic definitions of transactional memory (TM) 
%and recall the definition of correctness of an STM implementation.
Our STM model, while keeping the spirit of the original definitions of~\cite{GK08-opacity,tm-book}, 
introduces some refinements that are instrumental for our results. 

\paragraph{Transactions.} 
Transactional memory provides the ability of reading and writing to a set 
of \emph{transactional} objects, or \emph{t-objects} using atomic \emph{transactions}.
A transaction is a sequence of accesses (reads or writes) to t-objects.
We assume that every transaction $T_k$ has a unique identifier $k$. 
Formally, STM exports the following operations (called \emph{tm-operations} in the paper): 
(1) \textit{read}$_k(X)$ that returns a value in a set $V$ or a special value $A_k\notin V$ (\emph{abort});
(2) \textit{write}$_k(X,v)$ that returns \textit{ok}$_k$ or $A_k$;
(3) \textit{tryC}$_k$ that returns $C_k\notin V$ (\emph{commit})or $A_k$ and
(4) \textit{tryA}$_k$ that returns $A_k$. 

A {\it history} $H$ is a sequence of invocations and responses of tm-operations.
A history $H$ is \emph{sequential} if every invocation is either the last event in $H$ or 
is immediately followed by a matching response.  
%
%For a transaction $T_k$, 
$H|k$ denotes the subsequence of $H$ restricted to events with index $k$. 
If $H|k$ is non-empty we say that $T_k$ \emph{participates} in $H$, and \parts$(H)$ 
denotes the set of transactions that participate in $H$. 
A history is \emph{well-formed} if for all $T_k$, $H|k$ is sequential and 
contains no events that appear after $A_k$ or $C_k$.
Throughout this paper, we assume that all histories are well-formed, i.e., 
the user of transactional memory never invokes a new operation before receiving a response from 
the current one and does not invoke any operation $op_k$ after $T_k$ has returned $C_k$ or $A_k$.
A history $H$ is \emph{complete} if for every $T_k\in\textit{parts}(H)$,  
$H|k$ ends with a response event.
A transaction $T_k\in\textit{parts}(H)$ is \emph{live} in $H$ if $H|k$ does not 
end with $A_k$ or $C_k$.
Otherwise, $T_k$ is called \emph{complete}.
A history is \emph{t-complete} if $\parts(H)$ 
contains only complete transactions.
A transaction $T_k\in\parts(H)$ is \emph{forcefully aborted} 
in $H$ if some operation $op_k\neq \TryA_k$ returns $A_k$. 
Two histories $H$ and $H'$ are \emph{equivalent} if for every transaction $T_k$, $H|k=H'|k$. 

The \emph{read set} (resp., the \emph{write set}) 
of a transaction $T_k\in\parts(H)$, denoted  $\Rset(T_k)$ (resp., $\Wset(T_k)$), 
is the set of t-objects that $T_k$ reads (resp.,  writes to) in $H$. 
$\Dset(T_k)=\Rset(T_k)\cup\Wset(T_k)$ is called the \emph{data set} of $T_k$.  
A transaction $T_k$ is called \emph{read-only} if $\Wset(T_k)=\emptyset$, 
otherwise, it is called \emph{updating}.

\paragraph{Real-time and deferred-update orders.}

For $T_k,T_m\in\textit{parts}(H)$, we say that  $T_k$ {\it precedes} 
$T_m$ in the \emph{real-time order} in $H$, and we write $T_k\prec_H T_m$, if $T_k$ 
is committed or aborted and the last event of $T_k$ precedes the first event of $T_m$ in $H$.
If neither $T_k\prec_H T_m$ nor $T_m\prec_H T_k$, then we say that $T_k$ and $T_m$ 
are \emph{concurrent} in $H$.
A transaction $T_k\in\parts(H)$ which is not concurrent with any other transaction in $H$ 
is called \emph{uncontended} in $H$.  
A history $H$ is \emph{t-sequential} if no two transactions are concurrent in $H$.

For $T_k,T_m\in\textit{parts}(H)$, we say that  $T_k$ {\it precedes} 
$T_m$ in the \emph{deferred-update order}, 
and we write $T_k\prec_H^{DU} T_m$ if there exists
$X\in Rset(T_k)\cap Wset(T_m)$, $T_m$ has committed, such that the response
of \textit{read}$_k(X)$ precedes the invocation of $\TryC_m()$ in $H$.  
For $T_k,T_m\in\textit{parts}(H)$, we write $T_k {}_{\prec_{H}}^{X} T_m$, if $T_k$ has committed
and the response of \textit{read}$_m(X)$, $X \in \Rset(T_m)\cap \Wset(T_k)$ returns $v$, the
value of $X$ updated in \textit{write}$_k(X,v)$.

\paragraph{Legal histories.}

Let $H$ be a complete t-sequential history. 
For every operation $\Read_k(X)$ in $H$ that reads a t-object $X$, 
we define the \emph{latest written value} of $X$ as follows:
(1) If $T_k$ contains a $\Write_k(X,v)$ preceding $\Read_k(X)$ 
then the latest written value of $X$ is the value of the latest such write. 
(2) Otherwise, if $H$ contains a $\Write_m(X,v)$ such that $m\neq k$, 
$T_m$ precedes $T_k$, and $T_m$ commits in $H$, then
the latest written value of $X$ is the value of the latest such write in $H$. 
(3) Otherwise, the latest written value of $X$ is the initial value of $X$. 
Without loss of generality, we assume that  
$H$ starts with a fictitious initializing transaction $T_0$ 
that writes $0$ to every t-object. 
We say that a complete t-sequential history $H$ is \emph{legal}
if for every t-object $X$, every read of $X$ in $H$ 
returns the latest written value of $X$.     

\paragraph{Opacity.}

Let $H$ be any complete sequential history. 
Now $\bar H$ denotes a history constructed from $H$ as follows:
(1) For every live  transaction $T_k$ in $H$, we insert $\textit{tryC}_k\cdot A_k$ 
immediately after the last event of $T_k$ in $H$ and 
(2) For every aborted transaction $T_k$ in $H$, 
we remove all write operations in $T_k$ with the matching responses.  
\begin{definition}
\label{def:opacity}
A complete sequential history $H$ is \emph{opaque} if there exists 
a legal complete t-sequential history $S$ such that (1) $\bar H$ and $S$ are equivalent and
(2) $S$ respects $\prec_H$ and $\prec_H^{DU}$.
\end{definition}
We call such a legal complete t-sequential history $S$ a \emph{serialization} of $H$. 
A weaker property, called \emph{strict serializability}~\cite{Pap79-serial}, guarantees 
opacity with respect to committed transactions in $H$. 
Obviously, every opaque history is also strictly serializable. 

\paragraph{Implementations.}

We consider an asynchronous shared-memory system in which 
processes $p_1, \ldots p_N$ communicate by 
executing atomic operations on shared \emph{base objects}.

An STM \emph{implementation} provides the processes with algorithms for 
operations $\textit{read}_k$, $\textit{write}_k$, 
$\textit{tryC}_k$ and $\textit{tryA}_k$.
Without loss of generality, we assume that base objects are accessed 
with atomic read-write operations, 
%This allows us assume a \textit{total-order} relation on base object operations 
%abstracting the real-time order in which operations actually occur.
but we allow the programmer to aggregate a sequence of operations on base objects using
clearly demarcated \emph{atomic sections}: the operations within an atomic section are
to be executed sequentially.
%i.e., in each execution, events of an atomic section appear sequentially.
The atomic-section construct is general enough to implement various
strong synchronization primitives, such as \emph{test-and-set} (TAS)
or \emph{compare-and-swap} (CAS). 
We assume that atomic sections may only contain a bounded 
number of base-object operations.

An \emph{execution} of an implementation $M$ is a sequence of atomic accesses 
to base objects (\emph{base-object events}),  
%sequences of base-object events aggregated in atomic sections, 
and invocation and responses 
of the TM operations (\emph{TM-events}). 
If a base-object event is a write or an atomic-section that contains a write
(in one of its execution paths), we say that the event is \emph{non-trivial}.    

A \textit{configuration} of $M$ (after some execution $E$) 
is determined by the states of all base objects and the states 
of the processes.  
An \emph{initial state} of $M$ is determined by the initial states of base objects
and t-objects.
We assume that each base object and each t-object is initialized to $0$.
A \emph{history} of an execution $E$, denoted by $E|_{TM}$ is the subsequence of 
$E$ restricted to TM-events. $E|_{TM,p_i}$ denotes the subsequence 
of  $E|_{TM}$ restricted to events issued by process $p_i$. 

The \emph{interval of a transaction $T_k$ in $E$} is the fragment of $E$
that starts  with the first event of $T_k$ in $E$ and ends with 
the completing event of $T_k$ ($A_k$ or $C_k$) in $E$, or, if $T_k$ has
not completed in $E$,  with the last event of $E$. A tm-operation $op_1$ 
\emph{precedes} $op_2$ in $H$ if the invocation of $op_2$ appears after 
the response of $op_1$ in $H$. 
An execution $E$ is \emph{well-formed} if 
every atomic section is executed sequentially in $E$, 
$E|_{TM,p_i}$ is t-sequential for each $p_i$, 
and no event on behalf of a transaction $T_k$ is taking place outside 
of an interval between invocation and response of some TM-operation in $T_k$. 
We assume here that a TM implementation generates only 
well-formed executions.  

A \emph{completion} of $H$ is a history constructed from $H$ by removing
some pending invocations and adding responses to the remaining pending invocations 
to the end of $H$.  
To account for initial values of t-objects, we add to the beginning of $H$
a (fictitious) transaction $T_0$ that writes $0$ to every t-object and commits.  

A complete sequential history $H'$ is a \emph{linearization} of $H$ if there exists 
a history $H''$, a completion of $H$, such that (1)
$H'$ respects the precedence order of $H$, and (2) $H'$ and $H''$ are equivalent.
\begin{definition}
An STM implementation $M$ is \emph{opaque} if for every execution $E$ of $M$, 
there exists an opaque linearization of $E|_{TM}$. 
%that respects $\prec_H$ and $\prec_H^{DU}$.     
\end{definition}

%%%%%%%%%%%%%%%%%%%%%%%%%%%%%%%%%%%%
%\section{Preliminaries}
%\label{sec:prel}
%%%%%%%%%%%%%%%%%%%%%%%%%%%%%%%%%%%%
%%%%%%%%%%%%%%%%%%%%%%%%%%%%%%%%%%%%
\section{Preliminaries}
\label{sec:prel}
%%%%%%%%%%%%%%%%%%%%%%%%%%%%%%%%%%%%
In this section, we define some useful properties of STM implementations and
prove some simple facts that follow from these definitions.

\paragraph{Access patterns.}

The definition of STM allows a process to alternate reading and writing to t-objects arbitrarily 
in the course of a transaction. Moreover, it allows a process to read from a t-object 
that was previously written within the same transaction.
We show that this flexibility can be obtained ``for free'' given an implementation that 
only allows a user to read from a set of t-object and then to write to a set of t-objects
within a transaction.

We say that a transaction $T_k$ is \emph{canonic} in a history $H$ if $H|k$ consists 
is a sequence of reads (of distinct t-objects) followed by a sequence of writes 
(to distinct t-objects). 
A general \emph{complexity} of an STM implementation $M$ accounts for the number of accesses 
to base-objects used to implement every given transaction in every execution of $M$.

\begin{lemma}
\label{lem:access}
Let $M$ be an opaque STM implementation that can only be accessed with canonic transactions.
Then there exists an opaque STM implementation $M'$ that preserves the complexity of $M$.
\end{lemma}
\begin{proof}
Let $\Read^M$, $\Write^M$, $\TryC^M$ and $\TryA^M$ denote the implementations of the operations 
provided by $M$. Now $M'$ is constructed as follows. 

We associate every transaction $T_k$ with a local variable $\Wset(T_k)$ 
which contains, at any moment of time, the current write set of $T_k$ with 
the values to be written.

When $\Write_k(X,v)$ is invoked, $(X,v)$ is simply added to $\Wset(T_k)$ and all other 
entries of the form $(X,v')$ are removed from $\Wset(T_k)$.  
When $\Read_k(X)$ is invoked, we first check if $X$ is in $\Wset(T_k)$ and if so, we return 
the value stored in $\Wset(T_k)$.
Otherwise, we invoke $\Read_k^M(X)$ and return the obtained value. 

When $\TryC_k()$ is invoked, we first execute $\Write_k(X,v)$ for each $(X,v)\in\Wset_k$. 
Since for each $X$ there can be at most one entry of the form $(X,v)$, 
the order in which these operations are invoked does not matter.  
Also, since all invocations of $\Write_k$ succeed all invocations of $\Read_k$, 
the resulting sequence of invocations of $M$ on behalf of $T_k$ is a canonic transaction. 
Operation $\TryA_k()$ is implemented as $\TryA_k^M()$.

Since $M$ is opaque, the resulting implementation is also opaque: just use the serialization 
of the resulting history of $M$. 
Since the modifications of $M$ involve only local variables, 
the base-object complexity of $M'$ is the same 
as that of $M$.
\end{proof}
Therefore, in the rest of the paper, we only consider canonic transactions, which 
simplifies the analysis without sacrificing generality.

\paragraph{Disjoint-access parallelism (DAP).}

In STM implementations, it is considered important to allow transactions that are not
related through  their data sets that they access to execute
independently.

Let $I$ be a fragment of an execution $E$. 
Following~\cite{israeli-disjoint,AHM09}, 
we first define a \emph{conflict graph} which relates transactions that are live
in $I$.
Vertices of the graph represent t-objects. 
The vertices representing distinct t-objects $X$ and $Y$ are related with an edge if and only if
there is a  transaction $T$ such that $\{X,Y\}\subseteq\Dset(T)$ and 
the interval of $T$ overlaps with $I$ in $E$.

Two transactions $T_i$ and $T_j$ are \emph{disjoint-access in $E$} if 
there is no path between an item in $\Dset(T_i)$ and an item in
$\Dset(T_j)$ in the conflict graph of the minimal execution interval
containing the intervals of $T_i$ and $T_j$.    

Two transactions \emph{contend on a base-object $x$} in an execution if both of them access $x$ and
and one of these accesses is non-trivial. 

Two transactions \emph{concurrently contend on a base-object $x$} in an execution if both of them have pending events on $x$ in the same configuration and
and one of them is non-trivial. 
\begin{definition}\label{def:dap}
An STM implementation $M$ is \emph{disjoint-access parallel} (DAP) if, for
all executions $E$ of $M$, 
two processes executing $T_i$ and $T_j$ 
concurrently contend on the same base object in $E$ only if   
$T_i$ and $T_j$ are not disjoint-access.
\end{definition}
\begin{definition}\label{def:dap}
An STM implementation $M$ is \emph{strict disjoint-access parallel} (SDAP) if, for
all executions $E$ of $M$, 
two processes executing $T_i$  and  $T_j$ 
contend on the same base object in $E$ only if   
$T_i$ and $T_j$ have disjoint data sets.
\end{definition}
\begin{definition}
An STM implementation $M$ provides \emph{strict data partitioning} if
every t-object $X$ is associated with a set of base object $\beta(X)$
such that $\forall X\neq Y$, $\beta(X)\cap\beta(Y)=\emptyset$ and 
a transaction $T_i$ can access a base object in $\beta(X)$ only if
$X\in\Dset(T_i)$.
\end{definition}
Any STM that provides strict data partitioning is also disjoint-access
parallel (but not vice versa).

\paragraph{Invisible reads and single-version opacity.}

An STM implementation $M$ uses \emph{invisible reads} if no execution of a tm-read operation incurs 
a write on a base object. 

%Let $E$ be an execution of $M$.
%Let $B(i,E)$ denote the set of base objects accessed by events of $T_i$ in $E$. 

%An STM implementation $M$ is \emph{disjoint-access parallel} if for all executions $E$ of $M$, 
%$\textit{Dset}(T_i)\cap \textit{Dset}(T_j)=\emptyset$ implies
%$B(i,E)\cap B(j,E)=\emptyset$. 

Let $H$ be a sequential history. We say that $T_i$ precedes $T_j$ in $H$ in the \emph{single-version} 
order, and we write $T_i\prec_{H}^{SV} T_j$ if 
there exists $X\in \textit{Wset}(T_i)\cap \textit{Rset}(T_j)$ such that $\textit{tryC}_i$ precedes 
$\textit{read}_j(X)$ in $H$.    

A sequential history $H$ is \emph{single-version opaque} if there exists a legal t-sequential history 
$H'$ such that:
\begin{enumerate}
\item $\bar H$ and $H'$ are equivalent;
\item $H'$ respects $\prec_H$ and $\prec_H^{DU}$ and
\item $H'$ respects $\prec_H^{SV}$.
\end{enumerate}

Now an STM implementation $M$ is \emph{single-version opaque} if for every execution $E$ of $M$,  
there exists an opaque single-version  linearization of $E|_{TM}$. 
Intuitively a single-version opaque implementation is opaque and maintains 
exactly one copy of a t-object's state at any given moment.

%%%%%%%%%%%%%%%%%%%%%%%%%%%%%%%%%%%%
%\section{Liveness and Progress}
%\label{sec:liveness}
%%%%%%%%%%%%%%%%%%%%%%%%%%%%%%%%%%%%
%%%%%%%%%%%%%%%%%%%%%%%%%%%%%%%%%%%%
\section{Liveness and Progress}
\label{sec:liveness}
%%%%%%%%%%%%%%%%%%%%%%%%%%%%%%%%%%%%
To describe the conditions under which a TM implementation does something useful,
we need to address two orthogonal dimensions. 
First, we need to give a \emph{tm-liveness} property~\cite{AS85} 
that determines the conditions under which an individual tm-operation must return.
Second, we need to give a \emph{progress condition} 
that describes the cases in which a transaction must commit. 
\subsection{TM-liveness properties}
%A TM-liveness property specifies the conditions under which a TM-operation must return.
A TM implementation $M$ is \emph{wait-free} if in every infinite execution of $M$, 
each tm-operation returns in a finite number of its own steps, regardless of the 
behavior of concurrent transactions.
In other words, a wait-free individual tm-operation (read, write, tryC or tryA) cannot 
be delayed because of a concurrent operation.
The property can be very beneficial if executions of transactions are subject 
to unpredictable delays or failures. 
%However, the property comes with a cost: e.g., it cannot be
%implemented using atomic reads and writes~\cite{GK09-progressiveness}. 

In this paper, we do not assume failures: every operation is expected to take steps 
until it terminates. 
Moreover, we are interested in deriving inherent costs of implementing non-trivial concurrency in TM.  
Therefore, we assume a weaker default tm-liveness guarantee, that we call \emph{starvation-freedom}. 
A TM implementation $M$ is starvation-free in every infinite execution of $M$, each tm-operation 
eventually returns, 
assuming that no concurrent tm-operation stops indefinitely before returning.  
Starvation-freedom allows a tm-operation to be delayed only
by a concurrent tm-operation.  
%In the vein of classical non-blocking or lock-free implementations of concurrent data structures~\ref{Her91},  we may think of %other liveness properties, such as lock-freedom (at least one operation returns) or     
\subsection{Progress conditions}
A progress condition determines the scenarios in which a transaction
is allowed to abort. Technically, unlike tm-liveness, 
a progress condition is a safety property~\cite{AS85}, since it can 
be violated in a finite execution.
The simplest non-trivial progress property we consider in this paper
is \emph{single-lock progressiveness} that says that a transaction can
only abort if there is a concurrent transaction.
Clearly, an opaque single-lock TM can be implemented using 
any mutual exclusion algorithm~\cite{Ray86} 
with one critical section per transaction. Stronger progress conditions
allow some transactions to progress concurrently in some scenarios implying \emph{positive concurrency}\footnote{This does not include transactions that guarantee \emph{obstruction-freedom}~\cite{OFTM}}. 

\emph{Progressiveness} allows an implementation to abort a
transaction only in case of a conflict. 
Transactions $T_i,T_j$ \emph{conflict} in a history $H$ on a t-object $X$ if
$T_i$ and $T_j$ are concurrent in $H$, 
$X\in\Dset(T_i)\cap\Dset(T_j)$,  and $X\in\Wset(T_i)\cup\Wset(T_j)$.
\begin{definition}
\label{def:prog}
A TM implementation $M$ is \emph{(weakly) progressive} if for every history $H$ of $M$ and every transaction $T_i \in \parts(H)$ that is forcefully aborted, there exists a prefix $H'$ of $H$ and a transaction $T_k \in \parts(H')$ that is live in $H'$, such that $T_k$ and $T_i$ conflict in $H'$.
\end{definition}
The \emph{strong progressiveness} property~\cite{GK09-progressiveness} additionally requires that
in case of a set of transactions conflict on a single t-object at least one transaction commits. 
The formal definition is inspired from \cite{tm-book}. 

Let $CObj_H(T_i)$ denote the set of t-objects over which transaction $T_i \in \parts(H)$ conflicts with any other transaction in history $H$
i.e. $X \in CObj_H(T_i)$ if there exists a transaction $T_k \in \parts(H)$, $k\neq i$, such that $T_i$ conflicts with $T_k$ on $X$ in $H$. Then, 
$CObj_H(Q)=\{CObj_H(T_i) |\forall T_i \in Q\}$, denotes the union of sets $CObj_H(T_i)$ for all transactions in $Q$.

Let $CTrans(H)$ denote the set of non-empty subsets of $\parts(H)$ such that a set $Q$ is in $CTrans(H)$ if no transaction in $Q$ conflicts with a transaction not in $Q$.
\begin{definition}
\label{def:sprog}
A TM implementation $M$ is \emph{strongly progressive} if there does not exist any history $H$ of $M$ in which for every set $Q \in CTrans(H)$ 
of transactions such that $|CObj_{H}(Q)| \leq 1$, every transaction in $Q$ is forcefully aborted in $H$.
\end{definition}
But since the goal of this paper is to derive a lower bound, 
we consider weak progressive implementations (from now on---simply \emph{progressive}). 
  
Let $C$ be any correctness property, i.e., any safety property on TM histories~\cite{AS85}.
The following property guarantees that no transaction is forcefully aborted if there is a chance of to commit the transaction and preserve correctness.     
\begin{definition}
\label{def:perm}
A TM implementation $M$ is \emph{permissive with respect to $C$} if
for every history $H$ of $M$ such that $H$ ends with a response
$r_k$ and 
% to an operation $op_k$ executed by transaction $T_k \in \parts(H)$
% and any 
% there exists a history $H'$ constructed from $H$ by
replacing $r_k$ with some $r_k\neq A_k$ gives a history that satisfies $C$, 
we have $r_k \neq A_k$.
\end{definition}
Therefore, permissiveness does not allow a transaction to abort, unless committing 
it would violate the execution's correctness. 
In this paper, we consider TM implementations that are permissive with respect to opacity.  
Clearly, permissiveness with respect to opacity is strictly stronger than progressiveness: 
every permissive opaque implementation is also progressive opaque, but not vice versa. 

A transaction in a permissive opaque implementation can only be forcefully aborted if it  
tries to commit:
\begin{lemma}
\label{lem:perm-tryC}
Let a TM implementation $M$ be permissive with respect to opacity. 
If a transaction $T_i$ is forcefully aborted executing 
an operation $op_i$, then $op_i$ is $\TryC_i$.
\end{lemma}
\begin{proof}
Suppose, by contradiction, that there exists a history $H$ of $M$ such that 
some $op_i\in\{\Read_i,\Write_i\}$ executed within a transaction $T_i$ returns $A_i$.
Let $H_0$ be the shortest prefix of $H$ that ends just before $op_i$ returns.
By definition, $H_0$ is opaque and any history $H_0\cdot r_i$ where $r_i\neq A_i$ is not opaque.
Let $H_0'$ be the serialization of $H_0$.

If  $op_i$ is a write, then $H_0\cdot\ok_i$ is also opaque - no write operation of the incomplete transaction $T_i$ appears in $H_0'$ and, thus, $H_0'$ is also a serialization of $H_0\cdot\ok_i$. 

If  $op_i$ is a $\Read(X)$ for some t-object $X$, then we can construct a serialization of $H_0\cdot v$ where $v$ is the value of $X$ written by the last committed transaction in $H_0'$ preceding $T_i$ or the initial value of $X$ if there is no such transaction. 
It is easy to see that $H_0"$ obtained from $H_0'$ by adding $\Read(X)\cdot v$ at the end of $T_i$ is a serialization of 
$H_0\cdot\Read(X)$.
In both cases, there exists a non-$A_i$ response $r_i$ to $op_i$ that preserves opacity of  $H_0\cdot r_i$, and, thus,
the only operation that can be forcefully aborted in an execution of $M$ is $\TryC$.  
\end{proof}
Obviously, Lemma~\ref{lem:perm-tryC} implies that there does not exist a permissive single-version TM implementation.

\paragraph{Multi-version permissiveness.}

A relaxation of permissiveness, called~\emph{multi-version permissiveness} (or mv-permissiveness)~\cite{PFK10-mv} 
says that a transaction $T_i$ can only abort if $T_i$ is updating
and there is a concurrent conflicting updating transaction $T_j$ i.e. a read-only transaction cannot be aborted.  
\begin{lemma}
\label{lem:perm-sv}
There does not exist a mv-permissive TM implementation $M$ that guarantees (wait-freedom)starvation-freedom of individual tm-operations and single-version opacity.
\end{lemma}
\begin{proof}
By contradiction, suppose that there exists a single-version opaque mv-permissive implementation $M$.
Consider an execution of $M$ in which transaction $T_1$ sequentially reads $X$, then transaction $T_2$ writes to $X$ and $Y$ and commits. Such an execution exists, since none of these operations can be forcefully aborted in a mv-permissive implementation. Now we extend this history with $T_1$ reading $Y$.
There is no way to serialize $T_1$ and $T_2$ preserving single-version opacity, unless $\Read_1(Y)$ aborts.
But a mv-permissive TM implementation does not allow a read-only transaction to return abort--- a contradiction.       \end{proof}
If we relax our tm-liveness property and allow a tm-operation to be delayed 
by a concurrent conflicting transaction, then a single-version mv-permissive implementation 
is possible~\cite{AH11-perm}.  

\paragraph{Probabilistic permissiveness.}
Intuitively, a probabilistic permissive TM ensures the property of Definition~\ref{def:perm} \emph{with a positive probability}.
It is conjectured in~\cite{GHS08-permissiveness} that probabilistically
permissive (with respect to opacity) implementations can be considerably cheaper than deterministic ones. 
This is achieved by choosing the response to a tm-operation $op_k$ by sampling uniformly at random from the set of possible return values (including $A_k$).  
\begin{definition}
\label{def:pperm}
A TM implementation $M$ is \emph{permissive with respect to $C$} if
for every history $H$ of $M$ such that $H$ ends with a response
$r_k$ and 
% to an operation $op_k$ executed by transaction $T_k \in \parts(H)$
% and any 
% there exists a history $H'$ constructed from $H$ by
replacing $r_k$ with some $r_k\neq A_k$ gives a history that satisfies $C$, 
we have $r_k \neq A_k$ with positive probability.
\end{definition}
%This paper considers two  such properties: progress  \emph{progressiveness}~\cite{GK09-progressiveness} 
%and \emph{permissiveness}~\cite{GHS08-permissiveness}.
 
\section{RAW/AWAR complexity}
\label{sec:raw}
Modern CPU architectures perform reordering of memory references for better performance. Hence, memory barriers/fences are needed to enforce ordering in synchronization primitives whose correct operation depends on
ordered memory references. 
Attiya et al.~\cite{AGK11-popl} formalized the RAW/AWAR class of synchronization patterns and showed that a wide class of concurrent algorithm implementations must involve these expensive patterns. We recall the definitions below.

Let $\pi$ be an execution fragment and let $\pi_i$ denote the $i$-th
event in $\pi$ ($i=0,\dots,|\pi|-1$).   
We say that process $p$ performs a \emph{RAW} (read-after-write) in $\pi$ if  
$\exists i,j; 0 \leq i < j < |\pi|$
such that 
\begin{itemize}
\item $\pi_i$ is a write to a base object $x$ by process $p$,
\item $\pi_j$ is a read of a base object $y\neq x$ by process $p$ and 
\item there is no $\pi_k$ such that $i<k<j$ and $\pi_k$ is a write to
  $y$ by $p$.
\end{itemize}
We say that two RAWs by process $p$ \emph{overlap} in an execution $E$ with the read event of the
first RAW occurs after the write event of the second RAW. 
A \emph{multi-RAW} consists of series of writes to a set of base
objects followed by a series of reads from different base objects. 

We say a process $p$ performs an \emph{AWAR} (atomic-write-after-read)
in $\pi$ if $\exists i,j, 0 \leq i < j < |\pi|$ such that
\begin{itemize}
\item $\pi_i$ is a read of a base-object $x$ by process $p$,
\item $\pi_j$ is a write to a base-object $y$ by process $p$ and
\item $\pi_i$ and $\pi_j$ belong to the same atomic section.
\end{itemize}
Examples of AWAR are $\CAS$ and $\mCAS$.

%%%%%%%%%%%%%%%%%%%%%%%%%%%%%%%%%%%%
%\section{RAW complexity of permissive STMs}
%\label{sec:perm}
%%%%%%%%%%%%%%%%%%%%%%%%%%%%%%%%%%%%
%%%%%%%%%%%%%%%%%%%%%%%%%%%%%%%%%%%%
\section{RAW/AWAR cost of permissive STMs}
\label{sec:perm} 
In this section, we show that an execution of a transaction 
in a permissive STM implementation may require to perform 
at least one RAW/AWAR pattern \emph{per} tm-read.
 
Let $M$ be a permissive, opaque TM implementation.
Consider an execution $E$ of $M$ with a history $H$ consisting of 
transactions $T_1$, $T_2$, $T_3$ as shown in Figure~\ref{fig:H}:
$T_3$ performs a read of $X_1$, then $T_2$ performs a write on $X_1$ and commits, 
and finally $T_1$ performs a series of reads from objects $X_1,\ldots,X_m$.
Here, $R_k(X)$, $W_k(X,v)$ denote complete 
executions of $\Read_k(X)$ and $\Write_k(X,v)$ respectively.
%Note that $H=E|_{TM}$ is sequential.
%Let us ignore the subsequent write to $X_{k}$ performed by $T_3$ for now.
%Here $R_{k}(X)$  denotes the complete execution of $\Read_k(X)$,
%$W_{k}(X)$ denotes the complete execution of $\Write_k(X,1)$,
%and $TC_k$ denotes the complete execution $\TryC_k()$.
Since the implementation is permissive, no transaction can be
forcefully aborted in $E$, and 
the only valid serialization of this execution is $T_3$, $T_2$, $T_1$.  
Note also that the execution generates a sequential history: 
each invocation of a tm-operation is immediately followed by a matching response in $H$.
Thus, since we assume starvation-freedom as a liveness property, such an execution exists. 
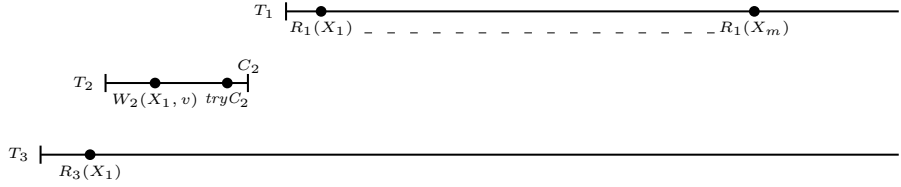
\begin{figure}[t]
\begin{center}
\scalebox{0.95}{
\begin{tikzpicture}
%%\node (b1) at (2,0) [fill,circle,inner sep=2.2pt] {};
%%\node (b2) at (-1,-1) [fill,circle,inner sep=2.2pt] {};
%%\node (b3) at (-1.8,-2) [fill,circle,inner sep=2.2pt] {};

\node (r1) at (2,0) [fill,circle,inner sep=1.5pt] {};
%\node (r2) at (3.5,0) [fill,circle,inner sep=1.5pt] {};
%%\node (e1) at (8,0) [fill,circle,inner sep=1.5pt] {};
\node (rm) at (8,0) [fill,circle,inner sep=1.5pt] {};
%\node (rj) at (6,0) [fill,circle,inner sep=1.5pt] {};
\node (w1) at (-0.3,-1) [fill,circle,inner sep=1.5pt] {};
\node (c2) at (0.7,-1) [fill,circle,inner sep=1.5pt] {};
%%\node (e2) at (1,-1) [fill,circle,inner sep=1.5pt] {};

\node (r3) at (-1.2,-2) [fill,circle,inner sep=1.5pt] {};
%\node (wk) at (9,-2) [fill,circle,inner sep=1.5pt] {};
%\node (c3) at (10,-2) [fill,circle,inner sep=1.5pt] {};
%%\node (e3) at (7,-2) [fill,circle,inner sep=1.5pt] {};
% Label the points

\draw (r1) node [below] {\tiny {$R_1(X_1)$}}
   (c2) node [below] {\tiny {$\TryC_2$}}
   %(c3) node [below] {\tiny {$TC_3$}}
   (rm) node [below] {\tiny {$R_1(X_{m})$}};
   
\draw (w1) node [below] {\tiny {$W_2(X_1,v)$}};
   
\draw 
(r3) node [below] {\tiny {$R_3(X_1)$}};
%   (w3) node [below] {\tiny {$W_3(X_2)$}}
%   (wk) node [below] {\tiny {$W_3(X_{k})$}};
%\begin{scope}
%\draw[loosely dashed] (2.5,-0.3) to (7.5,-0.3);
%\end{scope}  
\begin{scope}
\draw[loosely dashed] (2.6,-0.3) to (7.5,-0.3);
\end{scope}
\begin{scope}   
\draw [|-,thick] (1.5,0) node[left] {{\tiny $T_1$}} to (10,0);
\end{scope}
\begin{scope}
\draw [|-|,thick] (-1,-1) node[left] {{\tiny $T_2$}} to (1,-1) node[above] {{\tiny $C_2$}} ; 
\end{scope}  
\begin{scope}
\draw [|-,thick] (-1.9,-2) node[left] {{\tiny $T_3$}} to (10,-2) ; 
\end{scope} 
\end{tikzpicture}
}
\end{center}
\caption{\small {Execution $E$ of a permissive, opaque STM: $T_2$ and $T_3$ force $T_1$ to perform 
a RAW/AWAR in each $R_1(X_k)$, $2 \leq k \leq m$}}
\label{fig:H}
\end{figure}

Imagine that we modify the execution $E$ as follows. Immediately after $R_1(X_k)$ executed 
by $T_1$ we add $W_3(X,v)$, and $\TryC_3$ executed by $T_3$
(let $TC_3(X_k)$ denote the complete execution of $W_3(X_k,v)$ followed by $tryC_3$).
Obviously, $TC_3(X_k)$ must return abort: neither $T_3$ can be serialized before $T_1$ nor
$T_1$ can be serialized before $T_3$.
On the other hand if $TC_3(X_k)$ takes place just before $R_1(X_k)$, then 
$TC_3(X_k)$ must return commit but $R_1(X_k)$ must return the value written by $T_3$.
In other words, $R_1(X_k)$ and $TC_3(X_k)$ are \emph{strongly non-commutative}~\cite{AGK11-popl}:
both of them see the difference when ordered differently.    
As a result, intuitively, $R_1(X_k)$ 
needs to perform a RAW or AWAR to make sure that the order
of these two ``conflicting'' operations is properly maintained. A formal proof follows.
\begin{theorem}
\label{th:perm}
Let $M$ be a permissive opaque STM implementation.
Then, for any $m\in\Nat$, $M$ has an execution in which some transaction
performs $m$ tm-reads such that the execution of each tm-read contains 
at least one RAW or AWAR.
%contains   
%$\Omega(m)$ non-overlapping RAWs or AWARs on base objects where $m=|\Rset(T_i)|$.
\end{theorem}
\begin{proof}
We consider $R_1(X_k)$, $2 \leq k \leq m$ in execution $E$.
%\begin{enumerate}
%\item[(1)]

Imagine a modification $E'$ of $E$, in which $T_3$ performs $W_3(X_k)$ immediately after $R_1(X_k)$
and then tries to commit.
A serialization of $H'=E'|_{TM}$ should obey
$T_3\prec_{H'}^{DU} T_2$ and $T_2\prec_{H'} T_1$.
The execution of $R_1(X_k)$ does not modify base objects, hence, $T_3$
does not observe $R_1(X_k)$ in $E'$.
Since $M$ is permissive, $T_3$ must commit in $E'$. 
But since $T_1$ performs $R_1(X_k)$ before $T_3$ commits and $T_3$ 
updates $X_k$, we also have $T_1\prec_{H'}^{DU} T_3$.
Thus, $T_3$ cannot precede $T_1$ in any serialization---contradiction. 
Consequently, each $R_1(X_k)$ must perform a write to a base object.
%\item[(2)]
%Recall that $T_3$, $T_2$, $T_1$ is the only valid serialization for $E|_{TM}$.

Let $\pi$ be a fragment of $E$ that represents the complete execution of $R_1(X_k)$. 
Clearly, $\pi$ contains a write to a base object. Let $\pi_j$ be 
the first write to a base object in $\pi$ and $\pi_w$, the shortest fragment of $\pi$ 
that contains the atomic section to which $\pi_j$ belongs, else if $\pi_j$ 
is not part of an atomic section, $\pi_w=\pi_j$. Thus, $\pi$ can be represented 
as $\pi_s\cdot \pi_w\cdot \pi_f$.
 
Suppose that $\pi$ does not contain a RAW or AWAR. 
Since $\pi_w$ does not contain an AWAR, there are no read events 
in $\pi_w$ that precede $\pi_j$. Thus, $\pi_j$ is the first base object event in $\pi_w$. 
Consider the execution fragment $\pi_s\cdot \rho$, where $\rho$ is the complete execution 
of $TC_3(X_k)$ by $T_3$. Such an execution exists since $\pi_s$ 
does not perform any base object write, hence, $\pi_s \cdot\rho$ 
is indistinguishable to $T_3$ from $\rho$. 

Since, by our assumption, $\pi_w\cdot \pi_f$ contains no RAW, 
any read performed in $\pi_w\cdot \pi_f$ 
can only be applied to base objects previously written in $\pi_w\cdot \pi_f$. 
Thus, there exists an execution $\pi_s \cdot\rho \cdot \pi_w\cdot \pi_f$ 
that is indistinguishable to $T_1$ from $\pi$.
In $\pi_s \cdot\rho \cdot \pi_w\cdot \pi_f$, $T_3$ commits (as in $\rho$) but 
$T_1$ ignores the value written by $T_3$ to $X_k$.  
But $T_3$, $T_2$, $T_1$ is the only valid serialization for $E|_{TM}$---contradiction.
Thus, each $R_1(X_k)$, $2\leq k \leq m$ must contain a RAW/AWAR.
%\end{enumerate}

Note that since all tm-reads of $T_1$ are executed sequentially, all these RAW/AWAR patterns are
pairwise non-overlapping. 
\end{proof}

%%%%%%%%%%%%%%%%%%%%%%%%%%%%%%%%%%%%
%\section{Protected data in progressive STMs}
%\label{sec:prog}
%%%%%%%%%%%%%%%%%%%%%%%%%%%%%%%%%%%%
\section{RAW/AWAR cost and protected data in progressive STMs}
\label{sec:prog}
%
%\subsection{RAW/AWAR complexity of  a single-lock progressive STM}
%\label{sec:sl}
%
%
In this section, we first describe our progressive STM implementations 
that perform at most one RAW/AWAR per transaction.
Then we present a lower bound on the amount of data to be protected by a transaction in
a progressive STM.
\subsection{Constant RAW/AWAR implementations for progressive STM}
\label{sec:algs}
We start with showing that even a \emph{single-lock} progressive STM cannot avoid 
performing one RAW/AWARs per transaction in some executions.
\begin{theorem}
\label{th:sl}
Let $M$ be a single-lock progressive opaque STM implementation.
Then every execution of $M$ in which an uncontended transaction performs at least one 
read and at least one write contains a RAW/AWAR pattern.     
\end{theorem}
\begin{proof}
Consider an execution $\pi$ of $M$ in which an uncontended transaction $T_1$
performs (among other events) $\Read_1(X)$, $\Write_1(Y,v)$ and $\TryC_1()$.
Since $M$ is single-lock progressive, $T_1$ must commit in $\pi$. 
Clearly $\pi$ must contain a write to a base object. 
Otherwise a subsequent transaction reading $Y$ would return the
initial value of $Y$ instead of the value written by $T_1$.
 
Let $\pi_j$ be the first write to a base object in $\pi$ and let
$\pi_w$ denote the shortest fragment of $\pi$ that contains the atomic
section to which $\pi_j$ belongs ($\pi_w=\pi_j$ if $\pi_j$ is not
part of an atomic section). 
Thus, $\pi$ can be represented as $\pi_s\cdot\pi_w\cdot\pi_f$.

Now suppose, by contradiction, that $\pi$ contains neither RAW nor
AWAR patterns. 
Since $\pi_w$ contains no AWAR, there are no read events in $\pi_w$
that precede $\pi_j$. 
Since $\pi_j$ is the first write event in $\pi$, 
it follows that $\pi_j$ is the first base-object event in $\pi_w$.

%Consider  an execution extending $\pi_s$ in which another  transaction $T_2$.
Since $\pi_s$ contains no writes, the states of base objects in the
initial configuration and in the configuration after $\pi_s$ is performed are the same. 
Consider an execution $\pi_s\cdot\rho$ where in $\rho$, 
a transaction $T_2$ performs $\Read_2(Y)$, $\Write_2(X,1)$,
$\TryC_2()$ and commits.
Such an execution exists, since $\rho$ is indistinguishable to $T_2$
from an execution in which $T_2$ is uncontended and thus $T_2$ cannot
be forcefully aborted in $\pi_s\cdot\rho$. 

Since $\pi_w\cdot\pi_f$ contains no RAWs, every read performed in
$\pi_w\cdot\pi_f$ is applied to base objects which were previously
written in  $\pi_w\cdot\pi_f$. Thus, there exists an execution
$\pi_s\cdot\rho\cdot\pi_w\cdot\pi_f$,
such that $T_1$ cannot distinguish  $\pi_s\cdot\pi_w\cdot\pi_f$ 
and $\pi_s\cdot\rho\cdot\pi_w\cdot\pi_f$.
Hence, $T_1$  commits in  $\pi_s\cdot\rho\cdot\pi_w\cdot\pi_f$.

But both $T_1$ reads the initial value of $X$ and $T_2$ reads the initial value of $Y$ in
$\pi_s\cdot\rho\cdot\pi_w\cdot\pi_f$, and thus $T_1$ and $T_2$ cannot be both committed
(at least one of the committed transactions must read the value
written by the other)---a contradiction.

The proof is analogous in the case when an execution of $T_1$ 
extends any execution $\pi_0$ that contains only complete transactions.
\end{proof}
%
%Theorem~\ref{th:sl} shows that one RAW/AWAR is optimal for \emph{single-lock progressiveness}. 
%The proof is delegated to the Appendix (Section~\ref{app:sl}).
Since every progressive or permissive STM implementation is also single-lock progressive, 
the RAW/AWAR lower bound of Theorem~\ref{th:sl} also holds
for progressive and permissive STM implementations.
The lower bound is actually tight, and we sketch two progressive opaque
implementations.
Both implementations are \emph{strict data-partitioned}~\cite{tm-book} 
(split the set of base objects used into disjoint subsets, 
each subset storing information of only a single t-object) and 
\emph{single-version} (maintain exactly one copy of a t-object's state at a time).
They also use \emph{invisible reads}, i.e., no execution of a tm-read operation 
performs a write to a base object.

Our first implementation employs a \emph{mCAS} primitive\footnote{
In $\mCAS(V,OV,NV)$~\cite{AHCAS}, executed atomically, a process 
reads an array $V$ of $m$ objects $V$, and if for each $i$, $V[i]=OV[i]$, 
it replaces each $V[i]$ with $NV[i]$ and returns $\true$,
otherwise it returns $\false$ and leaves the objects unchanged.}
and works, in brief, as follows.
%[[PK Already gave the pointer once
%Algorithm~\ref{alg:cas} in Appendix describes the 
%pseudo-code of the implementation which uses a \emph{mCAS} primitive.
%.
%]]
Every t-object $X_i$ is associated 
with a distinct base object $v_i$ that stores the ``most recent'' 
value of $X_i$ together with the \emph{id} of the transaction that was the last to update $X_i$. 
Each time a transaction $T_k$ performs a read of a t-object $X_i$, 
it reads $v_i$, adds $X_i$ to its read set and checks if the t-objects in the current read set of $T_k$ 
have not been updated since $T_k$ has read them.
If this is not the case the transaction is forcefully aborted.
Otherwise, $T_k$ returns the value read in $v_i$.    
Each time $T_k$ performs a write to a t-object $X_i$, it adds $X_i$ to its  write set and returns \emph{ok}. 

For every updating transaction $T_k$, $\TryC_k()$ invokes 
the \emph{mCAS} primitive over $\Dset(T_k)$. 
If the \emph{mCAS} returns $\true$, $\TryC_k()$ returns $C_k$,
otherwise it returns $A_k$. Clearly, if $T_k$ is forcefully aborted,
then the execution of $\mCAS$ involved no AWAR (no write to a base object 
took place). Read-only transactions simply returns $C_k$.
Consequently, the implementation incurs a single AWAR per updating committed transaction.
\begin{theorem}
\label{th:kcas}
There exists a progressive opaque STM
implementation with wait-free operations that employs
exactly one AWAR per transaction.
Moreover, no AWARs are performed in read-only or aborted transactions. 
\end{theorem}
%The first implementation uses a single $\mCAS$ primitive per
%transaction and incurs a single AWAR per updating committed transaction. 
%
Even if we do not use atomic sections (and, thus, AWARs) we still can
implement a progressive opaque STM using reads and writes that incurs only a single multi-RAW (and,
thus, incurring just a single fence) per update transaction. 
This implementation uses a simple
\emph{multi-trylock} primitive, which in turn can be implemented with
a single multi-RAW. The $\multitrylock$ primitive
exports operations $\textit{acquire}(W)$, $\textit{release}(W)$ 
and $\textit{isContended}(X)$, for all sets  of t-objects $W$ and all t-objects $X$. 
Informally, if there is no contention on the locks on
objects in $W$, then $\textit{acquire}(W)$ returns $\true$ which means 
that exclusive locks on all objects in $W$ are acquired.
Otherwise, $\textit{acquire}(W)$ returns $\false$ which means that 
no locks on objects in $W$ are acquired.
Operation $\textit{release}(W)$ releases the acquired locks on objects in $W$
and $\textit{isContended}(X)$ returns $\true$ \emph{iff} a lock on $X$
is currently held by any process.  
The implementation of $\textit{acquire}(W)$
first writes to a series of base objects and then reads a series of base objects 
incurring a single multi-RAW,   
while operations $\textit{release}(W)$ and $\textit{isContended}(X)$ incur no RAW.

Implementations of reads and writes are similar to ones described above, except that
each time a transaction $T_k$ performs a read of a t-object $X_i$, 
%it reads $v_i$, adds $X_i$ to its read set and checks if the t-objects in the current read set of $T_k$ 
%have not been updated since $T_k$ has read them, and additionally 
it additionally checks if no object in the current read set is \emph{locked} by 
an updating transaction.
If some object in the read set has been modified or is 
locked, the transaction is forcefully aborted.
%Algorithm~\ref{alg:try} in the Appendix describes the pseudo-code of the implementation.

For every updating transaction $T_k$, $\TryC_k()$ invokes 
$\textit{acquire}(\Wset(T_k))$.
If it returns $\true$, $\TryC_k()$ returns $C_k$,
otherwise it returns $A_k$. Read-only transactions simply returns $C_k$. Consequently, the implementation
incurs a single multi-RAW per updating transaction.
\begin{theorem}
\label{th:mraw}
There exists a progressive opaque STM
implementation with wait-free operations that employs a
single multi-RAW per transaction.
Moreover, no RAWs are performed in read-only transactions. 
\end{theorem} 
%~\cite{GK09-progressiveness}.     
%\paragraph{Discussion.}
%Thus, in stark contrast to permissive implementations that contain a RAW or AWAR per tm-read (albeit in worst case executions), 
%we derive progressive STM implementations in which read-only transactions perform no RAW or AWAR. While permissive implementations provide maximum concurrency in theory, they require a non-trivial amount of expensive synchronization or powerful atomic constructs, both of which result in poor execution times. 
%
%
%
We also derive a \emph{strongly progressive} STM using only reads and writes that incurs at most four RAWs per updating transaction and uses a finite number of bounded registers.
Our implementation uses a \emph{starvation-free multi-trylock} primitive inspired by the Black-White Bakery Algorithm~\cite{Gadi-Bakery}, a bounded version of the Bakery Algorithm~\cite{Lamport74-bakery}. 

Informally, if no concurrent process contends infinitely long on some $X \in W$, then the $\textit{acquire}(W)$ operation of the starvation-free multi-trylock
eventually returns $\true$ which means 
that exclusive locks on all objects in $W$ are acquired.
%The $\textit{release}(W)$ and $\textit{isContended}(X)$ operations are analogous to the ones described in the constant RAW progressive implementation.
%Operation $\textit{release}(W)$ releases the acquired locks on objects in $W$
%and $\textit{isContended}(X)$ returns $\true$ \emph{iff} a lock on $X$
%is currently held by any process.  
The implementation of $\textit{acquire}(W)$
incurs three RAWs, while operation $\textit{release}(W)$ performs a single RAW. 
%Operation $\textit{isContended}(X)$ incurs no RAW.

Implementations of tm-reads and tm-writes are identical to the constant RAW progressive implementation described above. For every updating transaction $T_k$, $\TryC_k()$ invokes the $\textit{acquire}$ operation of the starvation-free multi-trylock over $\Wset(T_k)$. Note that this always returns $\true$
and a transaction $T_k$ with $\Rset_k=\emptyset$ eventually returns $C_k$.
Read-only transactions simply returns $C_k$. Consequently, the implementation
incurs four RAWs per updating transaction.
\begin{theorem}
\label{th:sprogop}
There exists a strongly progressive single-version opaque STM
implementation with starvation-free operations that uses invisible reads and employs
four RAWs per transaction.
Moreover, no RAWs are performed in read-only transactions. 
\end{theorem}
Note that our implementation does not violate the impossibility result of Guerraoui and Kapalka~\cite{tm-book} who proved that a strongly progressive opaque STM cannot
be implemented using only reads and writes if tm-operations are required to be \emph{wait-free}.
\subsection{Protected data}
\label{sec:prot}
%The constant RAW progressive implementation acquires the lock on the write set of the transaction, while the constant AWAR
%implementation uses an atomic-section construct over the data set of the transaction. This suggests that accounting for fences
%or strong synchronization primitives incurred by RAW/AWAR complexity may not be a good idea if we want to capture 
%the cost of progressiveness. To bridge this gap, we introduce a generic metric called \emph{protected data size} to characterize the cost
%of progressive STMs.
Let $M$ be a progressive STM implementation.
Intuitively, a t-object $X_j$ is protected at the end of some finite execution $\pi$ of $M$ 
if some transaction $T_0$ is about to atomically change the value of $X_j$ in its next step 
(e.g., by performing a CAS operation) or does not
allow any concurrent transaction to read $X_j$ (e.g., by holding a
lock on $X_j$).

Formally, let $\alpha\cdot\pi$ be an execution of $M$   
such that $\pi$ is an uncontended complete execution of 
a transaction $T_0$, where $\Wset(T_0)=\{X_1,\ldots,X_m\}$.
Let $u_j$ ($j=1,\ldots,m$) denote the value written by $T_0$ to t-object $X_j$ in $\pi$. 
We say that $\pi'$ is a \emph{proper prefix} of $\pi$ if $\pi'$ is a prefix of $\pi$ 
and every atomic section is complete in $\pi'$.
In this section, let $\pi^t$ denote the $t$-th shortest proper prefix of $\pi$. 
%$\pi_t$ denote the the $t$-th atomic event in $\pi$, and 
%$\pi^s,\pi^{s+1},\ldots,\pi^t$ ($1\le s<t\le|\pi|$) denote the fragment 
%of $\pi$ starting from $\pi^s$ and ending at $\pi^t$.
Let $\pi^0$ denote the empty prefix.
(Recall that an atomic event is either a tm-event, a read or write on a base object,
or an atomic section.)  

For any $X_j\in\Wset(T_0)$, %consider an extension of $\alpha\cdot\pi$ 
let $T_{j}$ denote a transaction that tries to read $X_j$ %(we assume that $m>k$)
and commit. % (we assume that $k>m$ and thus $T_0\notin\{T_1,\ldots,T_m\}$).
Let $E_{j}^{t}=\alpha\cdot\pi^t\cdot\rho_j^t$ denote the extension of $\alpha\cdot\pi^t$ in which    
$T_j$ runs solo until it completes.
Note that, since we only require the implementation to be starvation-free,  $\rho_j^t$ can be infinite. 

We say that $\alpha\cdot\pi^t$ is $(1,j)$-valent if the read
operation performed by $T_j$ in $\alpha\cdot\pi^t\cdot\rho_j^t$
returns $u_j$ (the value written by $T_0$ to $X_j$).
We say that $\alpha\cdot\pi^t$ is $(0,j)$-valent if the read operation performed by $T_j$ in $\alpha\cdot\pi^t\cdot\rho_j^t$
does not abort and returns an "old" value $u\neq u_j$.  
Otherwise, if the read operation of $T_j$  aborts or never returns in
$\alpha\cdot\pi^t\cdot\rho_j^t$, we say that  
$\alpha\cdot\pi^t$ is $(\bot,j)$-valent.
\begin{definition}
We say that $T_0$ \emph{protects} an object $X_j$ in
$\alpha\cdot\pi^t$, where $\pi^t$ is the $t$-th shortest proper prefix
of $\pi$ ($t>0$) if one of the following conditions holds:
%\begin{itemize}
%\item 
(1) $\alpha\cdot\pi^{t}$ is $(0,j)$-valent and $\alpha\cdot\pi^{t+1}$
is $(1,j)$-valent, or 
%\item 
(2) $\alpha\cdot\pi^{t}$ or  $\alpha\cdot\pi^{t+1}$ is $(\bot,j)$-valent.
%\end{itemize}
\end{definition}
%
%First, we observe that the no prefix of $\pi$ can be $0$ and $1$-valent at the same time. 
%By Lemma~\ref{lem:valence}, the notions of $0$-valence and $1$-valence
%are well-defined.    
%
%
%
For \emph{strict disjoint-access parallel} (SDAP) progressive STM, we show that 
every uncontended transaction must protect every object in its write set 
at some point of its execution.
%Theorem~\ref{th:protected} shows that every uncontended execution of a committed transaction 
%in an opaque \emph{disjoint-access parallel} (DAP) STM must pay an inherent price for 
%providing progressive concurrency by protecting its entire write set at some point 
%in its execution. 

We observe that the no prefix of $\pi$ can be $0$ and $1$-valent at the same time (notations used are the same as introduced in Section~\ref{sec:prot}). 
\begin{lemma}\label{lem:valence}
There does not exist $\pi^t$, a proper prefix of $\pi$, and
$i,j\in\{1,\ldots,m\}$ such that 
$\alpha\cdot\pi^t$ is both $(0,i)$-valent and  $(1,j)$-valent. 
\end{lemma}
\begin{proof}
By contradiction, suppose that there exist $i,j$ and $\alpha\cdot\pi^t$ 
that is both $(0,i)$-valent and  $(1,j)$-valent. 
Since the implementation is SDAP,
there exists an execution of $M$,
$E_{ij}^{t}=\alpha\cdot\pi^t\cdot\rho_j^t\cdot\rho_i^t$
that is indistinguishable to $T_i$ from
$\alpha\cdot\pi^t\cdot\rho_i^t$.
In $E_{ij}^{t}$, the only possible serialization is $T_0$, $T_j$, $T_i$.
But $T_i$ returns the ``old'' value of $X_i$ and, thus, the
serialization is not legal---a contradiction.
\end{proof}
If $\alpha\cdot\pi^t$ is $(0,i)$-valent (resp., $(1,i)$-valent) for
some $i$, we say that it is $0$-valent (resp., $1$-valent).
By Lemma~\ref{lem:valence}, the notions of $0$-valence and $1$-valence
are well-defined. \\
\\   
\begin{theorem}
\label{th:protected}
Let $M$ be a progressive, opaque and strict disjoint-access-parallel STM implementation.
Let $\alpha\cdot\pi$ be an execution of $M$, where $\pi$ is an
uncontended complete execution of a transaction $T_0$.
Then there exists $\pi^t$, a proper prefix of $\pi$, such that $T_0$
protects $|\Wset(T_0)|$ t-objects in $\alpha\cdot\pi^t$.  
\end{theorem}
\begin{proof}
Let $\Wset_{T_0}=\{X_1,\ldots,X_m\}$. 
Consider two cases:

\begin{enumerate}
\item[(1)]
Suppose that %for some $i,j\in\{1,\ldots,m\}$, 
$\pi$ has a proper prefix $\pi^{t}$ such that $\alpha\cdot\pi^t$ is
$0$-valent and $\alpha\cdot\pi^{t+1}$ is
$1$-valent.
By Lemma~\ref{lem:valence}, there does not exists $i$, such that $\alpha\cdot\pi^t$ is $(1,i)$-valent and $\alpha\cdot\pi^{t+1}$ is $(0,i)$-valent.
Thus, one of the following are true
\begin{itemize}
\item For every $i \in \{1,\ldots ,m\}$, $\alpha\cdot\pi^t$ is $(0,i)$-valent and $\alpha\cdot\pi^{t+1}$ is
$(1,i)$-valent
\item At least one of $\alpha\cdot\pi^t$ and $\alpha\cdot\pi^{t+1}$ is $(\bot , i)$-valent i.e. the operation of $T_i$ aborts or never returns
\end{itemize}
In either case, $T_0$ protects $m$ t-objects in $\alpha\cdot\pi^t$.  

\item[(2)] Now suppose that such $\pi^t$ does not exists, i.e., there
  is no $i\in\{1,\ldots,m\}$ and $t\in\{0,|\pi|-1\}$ such that $E_i^t$
  exists and returns an old value, and $E_i^{t+1}$ exists and returns
  a new value. 

Suppose there exists $s,t$, $0< s+1<t$,
$S\subseteq\{1,\ldots,m\}$,
such that:
\begin{itemize}
\item  $\alpha\cdot\pi^s$ is $0$-valent,  
\item  $\alpha\cdot\pi^t$ is $1$-valent,
\item  for all $r$, $s<r<t$, and for all $i\in S$, $\alpha\cdot\pi^r$ is $(\bot,i)$-valent.
\end{itemize}
We say that $s+1,\ldots,t-1$ is a \emph{protecting fragment} for
t-objects $\{X_j | j\in S\}$.  

Since $M$ is opaque and progressive, $\alpha\cdot\pi^0=\alpha$ is
$0$-valent and $\alpha\cdot\pi$ is $1$-valent.
Thus, the assumption of Case (2) implies that for each $X_i$, there exists a
protecting fragment for $\{X_i\}$.
In particular, there exists a protecting fragment for $\{X_1\}$.

Now we proceed by induction.
Let $\pi_{s+1},\ldots,\pi_{t-1}$ be a protecting fragment for
$\{X_1,\ldots,X_{u-1}\}$ such that $u\leq m$.
%Consider the protecting fragment for $\{X_{u}\}$. 

Now we claim that there must be a subfragment of $s+1,\ldots,t-1$ 
that protects $\{X_1,\ldots,X_u\}$. 

Suppose not. Thus, there exists $r$, $s<r<t$, such that 
$\alpha\cdot\pi^r$ is $(0,u)$-valent or $(1,u)$-valent. 
Suppose first that $\alpha\cdot\pi^r$ is $(1,u)$-valent.
Since $\alpha\cdot\pi^s$ is $(0,i)$-valent for some $i \neq u$, by Lemma~\ref{lem:valence}
and the assumption of Case (2),
there must exist $s',t'$, $s< s'+1 < t'\le r$ such that 
\begin{itemize}
\item  $\alpha\cdot\pi^{s'}$ is $0$-valent,  
\item  $\alpha\cdot\pi^{t'}$ is $1$-valent,
\item  for all $r'$, $s'<r'<t'$, $\alpha\cdot\pi^{r'}$ is $(\bot,u)$-valent.
\end{itemize}
As a result, $s'+1,\ldots,t'-1$ is a protecting fragment for $\{X_1,\ldots,X_u\}$.  
The case when $\alpha\cdot\pi^r$ is $(0,u)$-valent is symmetric, except that now we should 
consider fragment $r,\ldots,t$ instead of $s,\ldots,r$.

Thus, there exists a subfragment of ${s+1},\ldots,{t-1}$ that protects $\{X_1,\ldots,X_u\}$. 
By induction, we obtain a protecting fragment $s''+1,\ldots,t''-1$ for $\{X_1,\dots,X_m\}$.
Thus, any prefix $\alpha\cdot\pi^r$, where $s''<r<t''$ protects exactly
$m$ t-objects.  
\end{enumerate}

In both cases, there is a proper prefix of $\alpha\cdot\pi$ that protects
exactly $m$ t-objects.
\end{proof}
The lower bound of Theorem~\ref{th:protected} is tight: it is matched by all progressive implementations 
we are aware of, including ones in Section~\ref{sec:algs}. 
%%Any disjoint-access-parallel  single-lock STM
%single-lock opaque STM is also progressive:
%a transaction in a single-lock STM can only be forcefully aborted if it observes a
%concurrent transaction. In a DAP implementation, a transaction $T$ can
%only observe a concurrent transaction $T'$ if
%$\Dset(T)\cap\Dset(T')\neq\emptyset$.
%
Note that any DAP single-lock STM implementation automatically 
provides a stronger progress condition than just single-lock progressiveness.
A transaction $T$ in a DAP single-lock STM can only be forcefully aborted if it 
observes a concurrent transaction $T'$ such that
$\Dset(T)\cap\Dset(T')\neq\emptyset$. 
This is not very far from progressiveness, where $T$ may abort only if $T$ and $T'$ 
experience a write-write or write-read conflict on a t-object.
Thus, in the realm of DAP STM implementations, progressiveness 
is very close to the weakest non-trivial progress condition. 

%%%%%%%%%%%%%%%%%%%%%%%%%%%%%%%%%%%%
%\section{Related work}
%\label{sec:related}
%%%%%%%%%%%%%%%%%%%%%%%%%%%%%%%%%%%%
%\input{related}
%%%%%%%%%%%%%%%%%%%%%%%%%%%%%%%%%%%%
\section{Related work}
\label{sec:related}
%%%%%%%%%%%%%%%%%%%%%%%%%%%%%%%%%%%%
%Permissiveness as a progress condition for transactions aims to reduce the number of aborted transactions 
%by forcing transactions to not abort unless required for correctness. 
Crain et al.~\cite{michel-permissive} proved that a
permissive opaque TM implementation cannot maintain invisible reads, which
inspired the derivation of our lower bound on RAW/AWAR complexity in Section~\ref{sec:perm}.  
%They also introduced a relaxed correctness condition, {\it virtual world consistency} 
%that allows for a permissive implementation with invisible reads. 
%Virtual world consistency is similar in spirit to opacity except on the state 
%observable by aborted transactions, whose 
%reads are required to be consistent with respect to the transaction's causal past only. 
%Their result implies that wait-free read-only transactions in a (probabilistic) permissive 
%opaque implementation must perform a base object write for every item in the read set. 

The RAW/AWAR complexity for concurrent implementations was recently introduced in ~\cite{AGK11-popl}. 
The proofs of Theorems~\ref{th:perm} and \ref{th:sl} extend the arguments used
in~\cite{AGK11-popl} to the STM context. 

A related paper by Attiya et al.~\cite{AHM09} showed that every permissive 
strictly serializable and DAP TM in which every read-only transaction
must commit in a wait-free manner 
has an execution in which some read-only transaction $T_k$ performs at least 
$|\Dset(T_k)|$-1 base-object writes. 
In this paper we do not assume that a read operation must be wait-free
and we do not require disjoint-access parallelism.  
Also, we focus the number of RAW/AWAR patterns and not only
base-object writes. 
On the other hand, we consider a stronger correctness property (opacity).
Therefore, our lower bound in Section~\ref{sec:perm}
incomparable with the one of~\cite{AHM09}.

To establish the lower bound on t-objects that must be "protected" in an opaque, 
progressive TM (Section~\ref{sec:prot}),
we use the definition of disjoint-access parallelism introduced in \cite{AHM09}. 
Guerraoui and Kapalka~\cite{tm-book} considered a stronger 
version of DAP called \emph{strict data-partitioning}
to prove a linear lower bound on the number of steps performed by a successful 
read operation in a progressive, opaque TM that uses invisible reads.
Interestingly, the constant RAW/AWAR implementations of progressive, opaque TMs sketched 
in Section~\ref{sec:prog} are strict data-partitioned.
%%%%%%%%%%%%%%%%%%%%%%%%%%%%%%%%%%%%
%\section{Concluding remarks}
%\label{sec:conclusion}
%%%%%%%%%%%%%%%%%%%%%%%%%%%%%%%%%%%%
%\input{conclusion}
%%%%%%%%%%%%%%%%%%%%%%%%%%%%%%%%%%%%
\section{Concluding remarks}
\label{sec:conclusion}
%%%%%%%%%%%%%%%%%%%%%%%%%%%%%%%%%%%%
In this paper, we derived inherent costs of implementing STMs 
with non-trivial concurrency guarantees. 
At a high level, our results suggest that providing high degrees
of concurrency in STM may incur considerable unavoidable costs.
Our results give rise to many intriguing questions,
and we list some of them below. 

In this paper, we focused on progress conditions that provide positive concurrency, 
progressiveness and permissiveness.
%\footnote{In the full version of this paper~\cite{KR11}, we 
%extend our results to \emph{probabilistic} permissiveness~\cite{GHS08-permissiveness} 
%(a relaxation of permissiveness) and 
%\emph{strong} progressiveness (a restriction of progressiveness)~\cite{GK09-progressiveness}.}  
%and can be implemented using locks. 
The results do not apply to \emph{obstruction-free} STMs~~\cite{OFTM}
that only guarantee that a transaction commits if it eventually runs without contention.
Effectively, an obstruction-free STM provides zero concurrency, since progress 
is guaranteed only when one transaction is active at a time.
However, unlike single-lock implementations, it does allow overlapping 
transactions to make progress (one at a time).
Does this incur higher RAW/AWAR complexity?
%[[PK Delegate the reader to the TR 
%We know that strong progressiveness with wait-free tm-operations 
%requires stronger synchronization primitives (e.g., TAS or CAS) 
%than reads and writes~\cite{GK09-progressiveness,tm-book}.
%Since our implementation in Section~\ref{sec:prog} 
%exports starvation-free tm-operations,
%we conclude that strong progressiveness is strictly harder to implement 
%than weak progressiveness (in the computability sense) and incurs a strictly higher  
%RAW/AWAR complexity than their weakly progressive counterparts.
%]]

We cannot expect the lower bound of Theorem~\ref{th:protected} (the
protected-data size) to apply to non-DAP STMs, including trivial
ones that allow storing the state of the whole STM in one base object.
One way to avoid trivialities is to assume that a base object can
store information only about a constant number of t-objects (the
\emph{constant-size information} property in~\cite{GK08-opacity}) which can
potentially give asymptotically close results.  

We focused on implementations that allow a tm-operation
to be delayed only by concurrent operations performed by other transactions. 
Does relaxing the tm-liveness property by allowing a read operation to wait 
until a concurrent transaction terminates~\cite{AH11-perm} 
improve the RAW/AWAR complexity with respect to permissive implementations?
It is easy to see that the proof of our permissive lower bound 
(Theorem~\ref{th:perm}) does not work for this case.
But it is unclear a priori how this may affect the cost of  progressive
implementations. 
%The progressive STMs described in Section~\ref{sec:prog} 
%are actually wait-free for individual operations. 

Last but not least, the results of this paper assume opacity as a correctness property.
Recently, multiple relaxations of opacity were proposed~\cite{FGG09,AMT10,michel-permissive,AHM09}.
It would be very interesting to understand the concurrency benefits gained by such 
relaxed consistency conditions.
%Overall, we believe that this paper improves our understanding of inherent trade-offs between 
%the amount of concurrency a TM implementation is able to provide, the correctness guarantee 
%it exports, and the costs of implementing individual operations on transactional objects.
%%However, many interesting questions are yet to be answered.
%\newpage

\paragraph{Acknowledgements.}

The authors are grateful to Michel Raynal and Rachid Guerraoui for
inspiring discussions on the properties and costs of STM and Damien Imbs for valuable comments 
on the previous drafts. The comments and suggestions of anonymous reviewers on 
an earlier version of this paper are also gratefully acknowledged.
%\thispagestyle{empty}
%\clearpage
%
{\small 
\bibliography{references}
}
%\newpage
\appendix 
\newpage
\section{Constant RAW/AWAR implementations for progressive TM}
\label{app:prog}
This section presents the pseudo-code for single RAW and single AWAR implementations of progressive opaque STMs and their proofs of correctness.
% for Theorems~\ref{th:mraw} and \ref{th:kcas} omitted from Section~\ref{sec:algs}. 
The single RAW implementation
uses a \emph{multi-trylock} primitive described below, while the single AWAR implementation
uses a \emph{mCAS} primitive. Finally, we describe the read-write implementation of a strongly progressive STM that employs at most four RAWs per transaction. In the implementations, every t-object $X_i$ is associated 
with a distinct base object $v_i$ that stores the ``most recent'' 
value of $X_i$ together with the \emph{id} of the transaction that was the last to update $X_i$. 
\subsection{Multi-trylock}
\label{app:ml}
\begin{algorithm}[!ht]
\caption{Multi-trylock invoked by process $p_i$}\label{alg:mul}
%\caption{Multi-trylock implementation using one multi-RAW}\label{alg:try}
  \begin{algorithmic}[1]
  	\begin{multicols}{2}
  	{\footnotesize
	\Part{Shared variables}{
%		\State N=\# Processes
%		\State $x_i\forall i=1 \ldots M$~t-objects
		%\State $v_j$, for each t-object $X_j$
%   	\State $L[1, \ldots M]$~Multi Try-lock
		%\State $Q$-set of t-objects
	  	\State $r_{ij}$, for each process $p_i$ and each t-object $X_j$
		%\State $L$, a multi-trylock object
		%\State $L$, multi-trylock
	}\EndPart	
	%\Statex
	\Part{\textit{acquire}($Q$)}{
   		\ForAll{$X_j \in Q$}	
			\State \Write$(r_{ij},1)$
		\EndFor
		%\If{$\exists X_j \in Rset \cup Wset,\; t\neq i$: $\{r_{tj},r_{ij}\}\cap\{2\}\neq\emptyset$ 
		%		and $r_{tj}\neq 0$}
		%	\State \Write$(r_{ij},0)$
		%	\State return $\false$
		\If{$\exists X_j \in Q;t\neq i : r_{tj}=1$} \label{line:lock}
			\ForAll{$X_j \in Q$}	
				\State \Write$(r_{ij},0)$
			\EndFor
			\State return $\false$
		\EndIf
		\State return $\true$ 
	}\EndPart		
	 \Statex
	 \Part{\textit{release}($Q$)}{
  		\ForAll{$X_j \in Q$}	
 			\State \Write$(r_{ij},0)$
		\EndFor
		%\State $Rset_k:=Wset_k:=\emptyset$
		\State return \ok
	}\EndPart
	\Statex
	\Part{{\it \textit{isContended}($X_j$)}}{
		\If{$\exists p_t: r_{tj} \neq 0, t \neq i$} \label{line:isl}
			\State return $\true$
		\EndIf
		\State return $\false$
	}\EndPart			
	}\end{multicols}
  \end{algorithmic}
\end{algorithm}
A \textit{multi-trylock} provides exclusive write-access to a set $Q$ of t-objects. Specifically, a \textit{multi-trylock} exports the following operations
\begin{itemize}
\item \textit{acquire(Q)} returns {\it true} or {\it false}
\item \textit{release(Q)} releases the lock and returns \textit{ok}
\item \textit{isContended($X_j$)}, $X_j \in Q$ returns \emph{true} or \emph{false}
\end{itemize}
We assume that processes are well-formed: they never invoke a new
operation on the multi-trylock before receiving response from the
previous invocation. 

We say that a process $p_i$ \emph{holds a lock on $X_j$ after an execution $\pi$} if
$\pi$ contains the invocation of \textit{acquire($Q$)}, $X_j\in Q$ by
$p_i$ that returned \emph{true}, but does not contain a subsequent
invocation of \textit{release($Q'$)}, $X_j\in Q'$, by $p_i$ in $\pi$. 
We say that $X_j$ is \emph{locked after $\pi$} by process $p_i$ if $p_i$
holds a lock on $X_j$ after $\pi$.

We say that $X_j$ is \emph{contended by $p_i$ after an execution $\pi$} if
%there is a process $p_i$ such that 
$\pi$ contains the invocation of \textit{acquire($Q$)}, $X_j \in Q$,
by $p_i$ but does not
contain a subsequent return \emph{false} or return of
\textit{release($Q'$)}, $X_j \in Q'$, by $p_i$ in $\pi$. 
%Otherwise, we refer to $X_j$ as \emph{uncontended after $\pi$}.
 
Let an execution $\pi$ contain the invocation $i_{op}$ of an operation $op$
followed by a corresponding response $r_{op}$ (we say that $\pi$
\emph{contains} $op$). 
We say that \emph{$X_j$ is uncontended (resp., locked) during the execution of $op$ in
  $\pi$} if $X_j$ is uncontended (resp., locked) after every prefix of $\pi$ that
contains $i_{op}$ but does not contain $r_{op}$.
     
A multi-trylock implementation satisfies the following properties:
\begin{itemize}
\item \emph{Mutual-exclusion}: For any object $X_j$, and any execution
  $\pi$, there exists at most one process that \emph{holds} a lock on
  $X_j$ after $\pi$. 

\item \emph{Progress}:  Let $\pi$ be any execution that contains
  $\textit{acquire}(Q)$ by process $p_i$.
If no object in $Q$ is contended during the execution of
$\textit{acquire}(Q)$ by a process $p_k\neq p_i$ in $\pi$ then $\textit{acquire}(Q)$ returns $\true$ in $\pi$.

\item Let $\pi$ be any execution that contains $\textit{isContended}(X_j)$ invoked by $p_i$.

\begin{itemize}
\item If $X_j$ is locked by $p_{\ell};\ell \neq i$ during the complete execution of
$\textit{isContended}(X_j)$ in $\pi$, then
$\textit{isContended}(X_j)$ returns $\true$.

\item If $\forall \ell \neq i$, $X_j$ is never contended by $p_{\ell}$ during the complete execution of
$\textit{isContended}(X_j)$ in $\pi$, then
$\textit{isContended}(X_j)$ returns $\false$.

\end{itemize}     
Note that if $X_j$ is neither locked or uncontended during the
complete execution of  $\textit{isContended}(X_j)$, then either of $\true$ and
$\false$ can be returned.
\end{itemize}     
\begin{theorem}
\label{th:lock}
Algorithm~\ref{alg:mul} is an implementation of multi-trylock object
in which every operation is wait-free, every operation incurs at most one multi-RAW,
and \textit{isContended} involves no base-object writes
\end{theorem}
\begin{proof}
Denote by $L$ the shared object implemented by
Algorithm~\ref{alg:mul}. 
The operations exported by $L$ are \textit{wait-free} i.e. every operation returns a value to the invoking process after a finite number of its own steps. This follows from the fact that the implementation of \textit{acquire}, \textit{release} and \textit{isContended} described by Algorithm~\ref{alg:mul} contains no unbounded loops or waiting statements. 

Assume, by contradiction, that $L$ does not provide mutual-exclusion:
there exists an execution $\pi$ after which processes $p_i$ and $p_k$
\emph{hold} a lock on the same object, say $X_j$. 
In order to hold the lock on $X_j$, process $p_i$ writes $1$ to register
$r_{ij}$ and then checks if any other process $p_k$ has written $1$ to $r_{kj}$.  
Since the corresponding operation {\it acquire(Q)}, $X_j \in Q$
invoked by $p_i$ returns {\it true}, $p_i$ read $0$ in $r_{kj}$ in Line~\ref{line:lock}. 
%If $r_{kj}$ is $0$, then. % and $p_i$ holds the lock on $X_j$. 
But then $p_k$ also writes $1$ to $r_{kj}$ and later finds that
$r_{ij}$ is 1. 
This is because $p_k$ can write $1$ to $r_{kj}$ only after the read of
$r_{kj}$ returned $0$ to $p_i$ which is preceded by the write of $1$ to
$r_{ij}$. 
Hence, there exists an object $X_j$ such that $r_{ij}=1;i\neq k$, 
but the conditional in Line~\ref{line:lock} returns {\it true} to process $p_k$--- a contradiction. 

$L$ also ensures progress. 
This is trivial since some process $p_i$ wishing to hold a lock on
$X_j$ in an execution $\pi$ invokes \textit{acquire($Q$)}, $X_j \in Q$ which writes $1$ to register $r_{ij}$ 
and then checks if any other process $p_k$ has written to register
$r_{kj}$. If no other process contends on $X_j$ during the execution of \textit{acquire($Q$)},  
the conditional on Line~\ref{line:lock} returns {\it true} 
and respectively, \textit{acquire($Q$)} must return \emph{true}. 

Let $\pi$ be any execution that contains $\textit{isContended}(X_j)$
executed by $p_i$.
If no process contends on $X_j$ during the execution of
$\textit{isContended}(X_j)$ in $\pi$, $p_i$ finds $0$ in
$r_{tj}=0,~\forall t$ and the conditional in Line~\ref{line:isl} returns \emph{false}. 
However, if $X_j$ is locked  during the execution of
$\textit{isContended}(X_j)$ in $\pi$, at any point of the execution 
there exists $t$ such that $r_{tj}=1$. 
Thus,  the conditional in Line~\ref{line:isl} returns \emph{true}
and, respectively, 
$\textit{isContended}(X_j)$ must return $\true$.

The implementation of $\textit{isContended}(X_j)$ only reads base
objects. 
The implementation of $\textit{acquire}(Q)$ first writes to a series of base objects and then reads a series of base objects incurring a single multi-RAW.   
The implementation of $\textit{release}(Q)$ only writes to base objects.
\end{proof}

\subsection{Progressive implementation with single multi-RAW}
\label{app:raw}
Algorithm~\ref{alg:try} describes the algorithms for tm-operations of a progressive opaque STM incurring at most a single multi-RAW per transaction. 

Each time a transaction $T_k$ performs a read of a t-object $X_i$, 
it reads $v_i$, adds $X_i$ to its read set and checks if the t-objects in the current read set of $T_k$ 
have not been updated since $T_k$ has read them and additionally 
checks if no object in the current read set is \emph{locked} by 
an updating transaction.
If some object in the read set has been modified or is 
locked, the transaction is forcefully aborted.
Otherwise, $T_k$ returns the value read in $v_i$.    

Each time $T_k$ performs a write to a t-object $X_i$, it adds $X_i$ to its  write set and returns \emph{ok}. 

The implementation of $\TryC_k()$ uses the $\multitrylock$ primitive described in Section~\ref{app:ml}. For every updating transaction $T_k$, $\TryC_k()$ invokes $L.\textit{acquire}(\Wset(T_k))$, where $L$ denotes the multi-trylock implemented by Algorithm~\ref{alg:mul}.
If it returns $\true$, $\TryC_k()$ returns $C_k$,
otherwise it returns $A_k$. Read-only transactions simply returns $C_k$.
\begin{algorithm}[t]
\caption{Progressive STM with one multi-RAW: the implementation of $T_k$ executed by $p_i$}\label{alg:try}
%\caption{Multi-trylock implementation using one multi-RAW}\label{alg:try}
  \begin{algorithmic}[1]
  	\begin{multicols}{2}
  	{\footnotesize
	\Part{Shared variables}{
%		\State N=\# Processes
%		\State $x_i\forall i=1 \ldots M$~t-objects
		\State $v_j$, for each t-object $X_j$
%   	\State $L[1, \ldots M]$~Multi Try-lock
	  	%\State $r_{ij}$, for each process $p_i$ and each t-object $X_j$
		\State $L$, a multi-trylock object
		%\State $L$, multi-trylock
	}\EndPart	
	\Statex
	\Part{\Read$_k(X_j)$}{
		\State $\textit{ov}_j := \Read(v_j)$ \label{line:read2}
		\State $\Rset(T_k) := \Rset(T_k)\cup\{X_j\}$ \label{line:rset}
		\If{isAbortable()} \label{line:abort0}
				\State return $A_k$
		\EndIf
		%\If{$X_j \not\in \Rset_k$}
		%\State $\textit{ov}_j := \Read(v_j)$ \label{line:read2}
		%\If{$\textit{ov}_j \neq \Read(v_j)$} \label{line:eval}
		%	\State return $A_k$
		%\EndIf
		\State return the value of $\textit{ov}_j$
		%\EndIf
   	 }\EndPart
	\Statex
	\Part{\Write$_k(X_j,v)$}{
		\If{$X_j \not\in \Wset(T_k)$}
			\State $\textit{nv}_j := v$
			\State $\Wset(T_k) := \Wset(T_k)\cup\{X_j\}$
			\State return $\ok_k$
		\EndIf
   	}\EndPart
	\Statex
	\Part{\TryA$_k$()}{
		\State return $A_k$
	 }\EndPart
	\Statex
	\Statex
	\Statex
	\Statex
	\Statex
	\Statex
	\Part{\TryC$_k$()}{
		\If{$|\Wset(T_k)|= \emptyset$}
			\State return $C_k$ \label{line:return}
		\EndIf
		\State locked $:= L.\textit{acquire}(\Wset(T_k))$\label{line:acq} 
		\If{not locked} \label{line:abort2} 
	 		\State return $A_k$
	 	\EndIf
		\If{isAbortable()} \label{line:abort3}
			\State $L.\textit{release}(\Wset(T_k))$ 
			\State return $A_k$
		\EndIf
		\ForAll{$X_j \in \Wset(T_k)$}
	 		 \State $\Write(v_j,(\textit{nv}_j,k))$ \label{line:write}
	 	\EndFor		
		\State $L.\textit{release}(\Wset(T_k))$   		
   		\State return $C_k$
   	 }\EndPart		
	 \Statex
 	\Part{Function: {\bf isAbortable()}}{
		\If{$\exists X_j \in \Rset(T_k):L.\textit{isContended}(X_j)$}
			\State return $\true$
		\EndIf
		\If{isInvalid()} \label{line:valid}
			\State return $\true$
		\EndIf
		\State return $\false$
	}\EndPart
	\Statex	
	\Part{Function: {\bf isInvalid()}}{
		\If{$\exists X_j \in Rset(T_k)$:$\textit{ov}_j\neq \Read(v_j)$}
			\State return $\true$
		\EndIf
		\State return $\false$
	}\EndPart	
	}\end{multicols}
  \end{algorithmic}
\end{algorithm}
\subsubsection{Proof of opacity}
\label{sec:op1}
Let $E$ by any execution of the TM implemented by
Algorithm~\ref{alg:try}. 
Recall that we assume every t-object was initialized by some fictitious 
committed transaction $T_0$ that precedes $E$.
Let $<_E$ denote a total-order on events in $E$.

\paragraph{Linearization points.}

Let $H$ denote a linearization of $E|_{TM}$ constructed by selecting
\emph{linearization points} of tm-operations performed in $E|_{TM}$.
The linearization point of a tm-operation $op$, denoted as $\ell_{op}$ is associated with  
a base object event or a tm-event performed during 
the lifetime of $op$ using the following procedure. 

First, we obtain a completion of $E|_{TM}$ by removing some pending
invocations and adding responses to the remaining pending invocations
involving a transaction $T_k$ as follows:
\begin{itemize}
\item Every incomplete \emph{read$_k$}, \emph{write$_k$} or $\TryA_k$ operation is removed from $E|_{TM}$
\item For every pending \emph{tryC$_k$}, if some base object $v_j$ was  written (Line~\ref{line:write}), 
the response $C_K$ is added to the end of $E|_{TM}$, else $A_k$ is added to the end of $E|_{TM}$
\end{itemize}
Now a linearization $H$ of $E|_{TM}$ is obtained by associating linearization points to
tm-operations in the obtained completion of $E|_{TM}$ as follows:
\begin{itemize}
\item For every tm-read $op_k$ that returns a non-A$_k$ value, $\ell_{op_k}$ is chosen as the event in Line~\ref{line:read2} of Algorithm~\ref{alg:try}, else, $\ell_{op_k}$ is chosen as invocation event of $op_k$
\item For every tm-write or tm-abort  $op_k$ that returns, $\ell_{op_k}$ is chosen as the invocation event of $op_k$
\item For every $op_k=tryC_k$ that returns $C_k$ such that $\Wset(T_k)
  \neq \emptyset$, $\ell_{op_k}$ is associated with the successful
  acquisition of the lock on $\Wset(T_k)$ (at the end of Line~\ref{line:acq}), else if $op_k$ returns $A_k$, $\ell_{op_k}$ is associated with the invocation event of $op_k$
\item For every $op_k=tryC_k$ that returns $C_k$ such that $\Wset(T_k) = \emptyset$, $\ell_{op_k}$ is associated with Line~\ref{line:return}
\end{itemize}
$<_H$ denotes a total-order on tm-operations in the complete sequential history $H$.

\paragraph{Serialization points.} 

The serialization of a transaction $T_j$, denoted as $\delta_{T_j}$ is
associated with the linearization point of a tm-operation 
performed within the lifetime of the transaction.

We obtain a t-complete history ${\bar H}$ from $H$ as follows: 
\begin{itemize}
\item For every transaction $T_k$ in $H$ that is live, we insert $\textit{tryC}_k\cdot A_k$ immediately after the last event of $T_k$ in $H$. 
\item For every aborted transaction $T_k$ in $H$, we remove each write operation in $T_k$ with the matching response
\end{itemize}
${\bar H}$ is thus a t-complete sequential history that contains only updating committed transactions and read-only transactions since every aborted transaction is reduced to its read-prefix.
A serialization $S$ is obtained by associating serialization points to transactions in ${\bar H}$ as follows:
\begin{itemize}
\item If $T_k$ is an update transaction that commits, then $\delta_{T_k}$ is $\ell_{tryC_k}$
\item If $T_k$ is a read-only or aborted transaction,
then $\delta_{T_k}$ is assigned to the linearization point of the last tm-read that returned a non-A$_k$ value in $T_k$
\end{itemize}
$<_S$ denotes a total-order on transactions in the t-sequential history $S$.
\begin{lemma}
\label{lm:total}
If $T_i \prec_{H}^{DU} T_j$, then $T_i <_S T_j$
\end{lemma}
\begin{proof}
Recall that $T_i$ precedes $T_j$ in the \emph{deferred-update} order if there exists
$X\in Rset(T_i)\cap Wset(T_j)$, $T_j$ has committed, such that the response
of \textit{read}$_i(X)$ precedes the invocation of $tryC_j()$ in $H$. Thus, $\ell_{\textit{read}_i(X)} <_E \ell_{tryC_j}$.

Consider the histories depicted in Figure~\ref{fig:dag} where $T_i$
precedes $T_j$ in the deferred-update order ($tryC_k(X_j)$ 
denotes a \textit{tryC$_k$} such that $X_j \in \Wset(T_k)$).
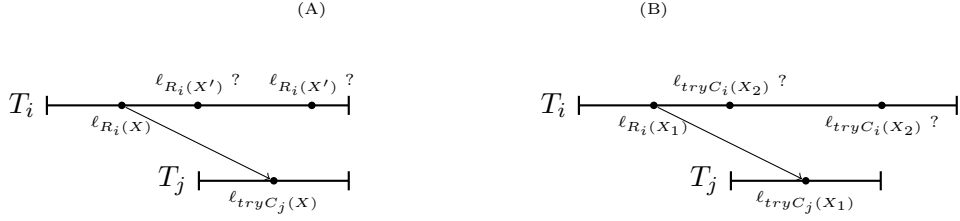
\begin{figure}
\begin{center}
\begin{tikzpicture}
%%\node (b1) at (2,0) [fill,circle,inner sep=2.2pt] {};
%%\node (b2) at (-1,-1) [fill,circle,inner sep=2.2pt] {};
%%\node (b3) at (-1.8,-2) [fill,circle,inner sep=2.2pt] {};
\node (l1) at (3.5,1.5) [] {};
\node (l2) at (8,1.5) [] {};

\node (r1) at (1,0) [fill,circle,inner sep=1pt] {};
\node (l3) at (2,0) [fill,circle,inner sep=1pt] {};
\node (l4) at (3.5,0) [fill,circle,inner sep=1pt] {};
%\node (r2) at (2.5,-1) [fill,circle,inner sep=1pt] {};
%\node (rk1) at (.5,-4) [fill,circle,inner sep=1pt] {};
%\node (rk) at (.5,-2) [fill,circle,inner sep=1pt] {};

%\node (c1) at (2.5,0) [fill,circle,inner sep=1pt] {};
\node (c2) at (3,-1) [fill,circle,inner sep=1pt] {};
%\node (ck1) at (2.5,-4) [fill,circle,inner sep=1pt] {};
%\node (ck) at (5.5,-2) [fill,circle,inner sep=1pt] {};

\draw (l1) node [below] {\tiny {(A)}};
\draw (l2) node [below] {\tiny {(B)}};

\draw (r1) node [below] {\tiny {$\ell_{R_i(X)}$}}
   (l3) node [above] {\tiny {$\ell_{R_i(X')}$ ?}}
	(l4) node [above] {\tiny {$\ell_{R_i(X')}$ ?}};
	%(rk) node [below] {\tiny {$R_3(X_3)$}};

\draw %(c1) node [below] {\tiny {$tryC_1(X_k)$}}
   (c2) node [below] {\tiny {$\ell_{tryC_j(X)}$}};
	%(ck1) node [below] {\tiny {$tryC_1(X_{k-2})$}}
	%(ck) node [below] {\tiny {$tryC_3(X_{2})$}};
      
\begin{scope}
\draw[->,thin] (r1) to (c2);
%\draw[->,thin] (r2) to (ck);
\end{scope}

\begin{scope}   
\draw [|-|,thick] (0,0) node[left] {$T_i$} to (4,0);
\end{scope}
\begin{scope}
\draw [|-|,thick] (2,-1) node[left] {$T_j$} to (4,-1); 
\end{scope}  
%\begin{scope}
%\draw [|<->|,thick] (3,-2) node[left] {$T_{3}$} to (6,-2) node[right] {$C_3$} ; 
%\end{scope} 
\node (r1) at (8,0) [fill,circle,inner sep=1pt] {};
%\node (r2) at (10.5,-1) [fill,circle,inner sep=1pt] {};
%\node (rk1) at (.5,-4) [fill,circle,inner sep=1pt] {};
%\node (rk) at (.5,-2) [fill,circle,inner sep=1pt] {};

\node (c11) at (9,0) [fill,circle,inner sep=1pt] {};
\node (c12) at (11,0) [fill,circle,inner sep=1pt] {};
\node (c2) at (10,-1) [fill,circle,inner sep=1pt] {};
%\node (ck1) at (2.5,-4) [fill,circle,inner sep=1pt] {};
%\node (ck) at (5.5,-2) [fill,circle,inner sep=1pt] {};

\draw (r1) node [below] {\tiny {$\ell_{R_i(X_1)}$}};
   %(r2) node [below] {\tiny {$R_2(X_2)$}};
	%(rk1) node [below] {\tiny {$R_1(X_{k-1})$}}
	%(rk) node [below] {\tiny {$R_3(X_3)$}};

\draw (c11) node [above] {\tiny {$\ell_{tryC_i(X_2)}$ ?}}
	(c12) node [below] {\tiny {$\ell_{tryC_i(X_2)}$ ?}}
   (c2) node [below] {\tiny {$\ell_{tryC_j(X_1)}$}};
	%(ck1) node [below] {\tiny {$tryC_1(X_{k-2})$}}
	%(ck) node [below] {\tiny {$tryC_3(X_{2})$}};
      
\begin{scope}
\draw[->,thin] (r1) to (c2);
%\draw[->,thin] (r2) to (c1);
\end{scope}
%\begin{scope}
%\draw[loosely dashed] (13.5,2) node[left] {\tiny {lock($X_2)=\delta_{T_1}$}} to (13.5,-2);
%\draw[loosely dashed] (14.5,2) node[below] {\tiny isLocked($X_1$)} to (14.5,-2);
%\draw[loosely dashed] (15.5,2) to (15.5,-2) node[right] {\tiny {lock$(X_1)=\delta_{T_2}$}};
%\end{scope}
\begin{scope}   
\draw [|-|,thick] (7,0) node[left] {$T_i$} to (12,0);
\end{scope}
\begin{scope}
\draw [|-|,thick] (9,-1) node[left] {$T_j$} to (11,-1); 
\end{scope}  
%\begin{scope}
%\draw [|<->|,thick] (3,-2) node[left] {$T_{3}$} to (6,-2) node[right] {$C_3$} ; 
%\end{scope} 
\end{tikzpicture}
\end{center}
\caption{Assignment of serialization points respects the deferred-update order}
\label{fig:dag}
\end{figure}
\begin{enumerate}
\item[(1)]
Consider the history depicted in Figure~\ref{fig:dag}(A) where $T_i$
is a read-only transaction and $T_j$ 
is an updating transaction that returns $C_j$. 
Assume the contrary that $T_i \prec_{H}^{DU} T_j$, but $T_j <_S T_i$,
which implies that 
$\delta_{T_j} <_E \delta_{T_i}$ i.e. $\ell_{tryC_j(X)}$ precedes the
linearization point 
of the last tm-read in $T_i$ that returns a non-A$_i$ value (say
\textit{read}$_i(X')$). 
Thus, successful lock acquisition on $X$ by $T_j$ in
Line~\ref{line:acq} 
precedes the read of the base object associated with $X'$ by $T_i$ in Line~\ref{line:read2}. 

\textit{read}$_i(X')$ checks if any object in $\Rset(T_i)$ is
locked by a concurrent transaction, then performs read-validation (Line~\ref{line:abort0}). 
Consider the following possible sequence of events: $T_j$ acquires the
lock on $X$, updates $X$ to shared-memory, 
$T_i$ reads the base object associated with $X'$, $T_j$ releases the
lock and finally $T_i$ performs 
the check in Line~\ref{line:abort0}. $read_i(X')$ is forced to return $A_i$ because $X$ has been invalidated. 

Else if $T_j$ acquires the lock on $X$, updates $X$ to shared-memory,
$T_i$ reads the base object associated with $X'$, 
$T_i$ performs the check in Line~\ref{line:abort0} and finally $T_j$
releases the lock on $X$. Again, $read_i(X')$ returns $A_i$ since
$T_j$ is holding a lock on $X \in Rset(T_i)$---contradiction.

Hence, the only possibility is that the last successful tm-read ($read_i(X')$) in $T_i$ is linearized before $tryC_j(X)$, which implies that $\delta_{T_i} <_E \delta_{T_j}$.

\item[(2)]
Suppose that $T_i$ is an updating transaction as shown in Figure~\ref{fig:dag}(B), then $\ell_{tryC_i(X_2)}$ and $\ell_{tryC_j(X_1)}$ are assigned to Line~\ref{line:acq} of Algorithm~\ref{alg:try} when the locks are acquired on $X_2$ and $X_1$ respectively. Assume the contrary that $T_i$ precedes $T_j$ in deferred-update order, but $\delta_{T_j} <_E \delta_{T_i}$, then $\ell_{tryC_j} <_E \ell_{tryC_i}$. A similar argument to the above leads to a contradiction since \emph{tryC} performs the same sequence of checks as the tm-read (Line~\ref{line:abort3}).
\end{enumerate}
\end{proof}
\begin{lemma}
\label{lm:seq}
If $T_i \prec_{H}T_j$, then $T_i <_S T_j$
\end{lemma}
\begin{proof}
This follows from the fact that for a given transaction, its
serialization point is chosen within the lifetime of the transaction
implying
if $T_i \prec_{H} T_j$, then $\delta_{T_i} <_{E} \delta_{T_j}$ $\Longrightarrow$ $T_i <_S T_j$ 
\end{proof}
%\begin{comment}
\begin{lemma}
\label{lm:serial}
If $T_i {}_{\prec_{H}}^{X} T_j$, then $T_i <_S T_j$
\end{lemma}
\begin{proof}
Assume the contrary, i.e. there exists a \textit{read}$_j(X)$, $X \in
Rset(T_j)\cap Wset(T_i)$ that returns the value of $X$ updated in 
\textit{write}$_i(X,value)$ and $T_j <_S T_i$. $T_i$ is an updating
committing transaction, hence $\delta_{T_i}=\ell_{tryC_i}$. 

Consider two cases:
\begin{enumerate}

\item[(1)]
Suppose that $T_j$ is a read-only transaction. Thus, $\delta_{T_j}$ is
assigned to the last tm-read that returns a non-A$_j$ value (say
$read_j(X')$), 
whose linearization point precedes $\ell_{tryC_i}$. This implies that
the read of the base object associated with $X'$ by $T_j$ in
Line~\ref{line:read2} 
precedes the successful lock acquisition on $X$ by $T_i$ in Line~\ref{line:acq}. 
Thus, the write to the base object associated with $X$ performed by
$\TryC_i()$ in line~\ref{line:write} is executed \emph{after} the read of
the base object performed by $\Read_j(X)$ in Line~\ref{line:read2}---a contradiction.

\item[(2)]
Suppose that $T_j$ is an updating transaction. Then, $\ell_{tryC_j}
<_E \ell_{tryC_i}$. 
Again, this implies that the read of the base object in
Line~\ref{line:read2} executed by $\Read_j(X)$ precedes to the write
to the base object performed by $\TryC_i()$---a contradiction.
\end{enumerate}
\end{proof}
\begin{lemma}
\label{lm:readfrom}
$S$ is legal
\end{lemma}
\begin{proof}
Recall that $S$ is legal if every tm-read of an object $X$ performed
by a transaction $T_i$ returns the response of the latest value
written to $X$ in $S$.
Since we only consider canonic transactions,  the latest value written
to $X$ in $S$ is the value written by the last transaction $T_j$ such
that $T_k$ commits, $T_j<_S T_i$ and $X\in \Wset(T_j)$.

From Lemma~\ref{lm:serial}, we have that for all $T_i$ and $T_j$, if $T_i {}_{\prec_{H}}^{X} T_j$, then
$T_i$ precedes $T_j$ in $S$.
Thus, to prove that $S$ is legal,  it is enough to show that  
if $T_i {}_{\prec_{H}}^{X} T_j$, then there does not exist a
transaction 
$T_k$ that returns $C_k$, $X \in \Wset(T_k)$ such that $T_i <_S T_k <_S T_j$. 

Assume the contrary that 
\begin{itemize}
\item $T_i {}_{\prec_{H}}^{X} T_j$
\item $\exists T_k$, $X \in \Wset(T_k)$,  returns $C_k$ such that $T_i <_S T_k <_S T_j$
\end{itemize}
$T_i$ and $T_k$ are both updating transactions that commit. Thus, 
\begin{center}
($T_i <_S T_k$) $\Longleftrightarrow$ ($\delta_{T_i} <_{E} \delta_{T_k}$) \\
($\delta_{T_i} <_{E} \delta_{T_k}$) $\Longleftrightarrow$ ($\ell_{tryC_i} <_{E} \ell_{tryC_k}$) 
\end{center}
Since, $T_j$ reads the value of $X$ written by $T_i$, one of the following is true 
\begin{center}
$\ell_{tryC_i} <_{E} \ell_{tryC_k} <_{E} \ell_{read_j(X)}$ (or) \\
$\ell_{tryC_i} <_{E} \ell_{read_j(X)} <_{E} \ell_{tryC_k}$
\end{center}
If $\ell_{tryC_k} <_{E} \ell_{read_j(X)}$, then the successful lock acquisition on $X$ by $T_k$ in Line~\ref{line:acq} precedes the read of the base object associated with $X$ by $T_j$ in Line~\ref{line:read2}. 

\textit{read}$_j(X)$ checks if any object in $\Rset(T_j)$ is locked by a concurrent transaction, then performs read-validation (Line~\ref{line:abort0}). Consider the following possible sequence of events: $T_k$ acquires the lock on $X$, updates $X$ to shared-memory, $T_j$ reads the base object associated with $X$, $T_k$ releases the lock and finally $T_j$ performs the check in Line~\ref{line:abort0}. $read_j(X)$ is forced to return $A_j$ because $X \in \Rset(T_j)$ (Line~\ref{line:rset}) and has been invalidated since last reading its value. 

Else if, $T_k$ acquires the lock on $X$, updates $X$ to shared-memory, $T_j$ reads the base object associated with $X$, $T_j$ performs the check in Line~\ref{line:abort0} and finally $T_k$ releases the lock on $X$. Again, $read_j(X)$ returns $A_j$ since $T_k$ is holding a lock on $X \in Rset(T_j)$---contradiction.

Thus, $\ell_{read_j(X)} <_{E} \ell_{tryC_k}$.

Consider two cases:
\begin{enumerate}
\item[(1)]
Suppose that $T_j$ is a read-only transaction. Then, $\delta_{T_j}$ is assigned to the last tm-read performed by $T_j$ that returns a non-A$_j$ value. If $read_j(X)$ is not the last tm-read that returned a non-A$_j$ value, then there exists a $read_j(X')$ such that 
\begin{center}
$\ell_{read_j(X)} <_{E} \ell_{tryC_k} <_E \ell_{read_j(X')}$
\end{center}
\item[(2)]
Suppose that $T_j$ is an updating transaction that commits, then $\delta_{T_j}=\ell_{tryC_j}$ which implies that
\begin{center}
$\ell_{read_j(X)} <_{E} \ell_{tryC_k} <_E \ell_{tryC_j}$
\end{center}
\end{enumerate}
The same argument derived in the proof of Lemma~\ref{lm:total} shows
that both cases lead to a contradiction, i.e., both $read_j(X')$ and
$tryC_j$ 
are forced to return $A_j$---contradiction.     
\end{proof}
\begin{lemma}
\label{lm:progtm}
Algorithm~\ref{alg:try} implements a progressive TM
\end{lemma}
\begin{proof}
Every transaction $T_k$ in a TM $M$ whose tm-operations are defined by Algorithm~\ref{alg:try} can be aborted in the following scenarios
\begin{itemize}
\item Read-validation failed in \textit{$read_k$} or \textit{tryC$_k$}
\item \textit{$read_k$} or \textit{tryC$_k$} returned $A_k$ because $X_j \in \Rset(T_k)$ is locked (belongs to write set of a concurrent transaction)  
\item \textit{L.acquire($\Wset(T_k)$)} returned \textit{false} in Line~\ref{line:abort2} of Algorithm~\ref{alg:try}
\end{itemize}
Read-validation consists of checking whether the value to be returned from a tm-read of transaction $T_k$ is consistent with the values returned from the previous tm-reads. Hence, if validation of a tm-read in $T_k$ fails, it means that the t-object is overwritten by some transaction $T_i$ such that $T_i <_S T_k$, implying a {\it read-write conflict}. This is also implied if some t-object $X_j \in \Rset(T_k)$ is locked and returns \emph{abort} since the t-object is in the write set of a concurrent transaction.

Acquisition of the multi-trylock can return \textit{false} for $T_i$ because there exists some $X_j \in \Wset(T_i)$ that was being written to by a concurrent transaction $T_k$ implying a \textit{write-write conflict}.
    
Hence, for every transaction $T_i \in H$ that is aborted, there exists a conflicting t-object that is contended by a concurrent transaction. Thus, Algorithm~\ref{alg:try} implements a progressive TM 
\end{proof}\\
\\
\begin{reptheorem}[Theorem~\ref{th:mraw}]
%\label{th:raw}
There exists a progressive opaque STM
implementation that employs a
single multi-RAW per transaction.
Moreover, no RAWs are performed by read-only transactions. 
\end{reptheorem}\\
\\
\begin{proof}
From Lemmas~\ref{lm:total}, \ref{lm:seq}, \ref{lm:readfrom} and \ref{lm:progtm}, Algorithm~\ref{alg:try} implements a progressive, opaque STM.
%Algorithm~\ref{alg:try} uses \emph{invisible reads} since the tm-read operation does not perform any base-object writes.

Any process executing a transaction $T_k$ holds the lock on $\Wset(T_k)$ only once during \emph{tryC$_k$}. If
$|\Wset(T_k)|=\emptyset$, then the transaction simply returns $C_k$
incurring no RAW's. Thus, from Theorem~\ref{th:lock}, Algorithm~\ref{alg:try} incurs a single multi-RAW per updating transaction and no RAW's are performed in read-only transactions.
\end{proof}
\subsection{Progressive implementation with single mCAS}
\label{app:mcas}
\begin{algorithm}[t]
\caption{Progressive STM with single mCAS; implementation of transaction $T_k$ by process $p_i$}\label{alg:cas}
%\caption{Multi-trylock implementation using one multi-RAW}\label{alg:try}
  \begin{algorithmic}[1]
  	\begin{multicols}{2}
  	{\footnotesize
	\Part{Shared variables}{
%		\State N=\# Processes
%		\State $x_i\forall i=1 \ldots M$~t-objects
		\State $v_j$, for each t-object $X_j$
		%\State $K=|Dset_k|$
%   	\State $L[1, \ldots M]$~Multi Try-lock
	  	%\State $r_{ij}$, for each process $p_i$ and each t-object $X_j$
		%\State $L$, multi-trylock
	}\EndPart	
	\Statex
	\Part{\Read$_k(X_j)$}{
		\State $\textit{ov}_j := \Read(v_j)$ \label{line:casread}
		\State $\Rset(T_k) := \Rset(T_k)\cup\{X_j\}$
		\State $\textit{nv}_j := \textit{ov}_j$ \label{line:copy}
		\If{isInvalid()} 
				\State return $A_k$
		\EndIf
		%\If{$X_j \not\in \Rset_k$}
		%\State $\textit{ov}_j := \Read(v_j)$ \label{line:read2}
		%\If{$\textit{ov}_j \neq \Read(v_j)$} 
		%	\State return $A_k$
		%\EndIf	
		
		\State return the value of $\textit{ov}_j$
		%\EndIf
   	 }\EndPart
	\Statex
	\Part{\Write$_k(X_j,v)$}{
		\If{$X_j \not\in \Wset(T_k)$}
			\State $\textit{nv}_j := v$ \label{line:copy2}
			\State $\Wset(T_k) := \Wset(T_k)\cup\{X_j\}$
			\State return $\ok_k$
		\EndIf
   	}\EndPart
	\Statex
	\Part{\TryA$_k$()}{
		\State return $A_k$
	}\EndPart	
	\Statex
	\Part{\TryC$_k$()}{
		\If{$\Wset(T_k)=\emptyset$}
			\State return $C_k$ \label{line:rcas}
		\EndIf
		\ForAll{$X_j \in \Wset(T_k)$}
	 		 \State $\textit{ov}_j := \Read(v_j)$
	 	\EndFor	
		%\Statex
		\State Let $\Wset(T_k) \cup \Rset(T_k)~be~\{X_{i_1},...,X_{i_m}\}$
		%\Statex
		\State $V = \{v_{i_1},...,v_{i_m}\}$
		%\Statex
		\State $OV = \{ov_{i_1},...,ov_{i_m}\}$
		%\Statex
		\State $NV = \{nv_{i_1},...,nv_{i_m}\}$
		%\Statex
		\If{mCAS(V,OV,NV)} \label{line:cas}
			\State return $C_k$
		\EndIf
		\State return $A_k$
   	 }\EndPart		
	\Statex
	\Part{Function: {\bf isInvalid()}}{
		\If{$\exists X_j \in Rset(T_k)$:$\textit{ov}_j\neq \Read(v_j)$}
			\State return $\true$
		\EndIf
		\State return $\false$
	}\EndPart		 
	}\end{multicols}
  \end{algorithmic}
\end{algorithm}
Algorithm~\ref{alg:cas} describes the implementation of a progressive, opaque TM incurring a single AWAR per updating committed transaction. The implementations of reads and writes are similar to ones described in Section~\ref{app:raw} except that each time a transaction $T_k$ performs a read of a t-object $X_i$, 
it reads $v_i$, adds $X_i$ to its read set and checks if the t-objects in the current read set of $T_k$ 
have not been updated since $T_k$ has read them.
If this is not the case, the transaction is forcefully aborted.
Otherwise, $T_k$ returns the value read in $v_i$.    

For every updating transaction $T_k$, $\TryC_k()$ invokes 
the \emph{mCAS} primitive over $\Dset(T_k)$. 
If the \emph{mCAS} returns $\true$, $\TryC_k()$ returns $C_k$,
otherwise it returns $A_k$. Read-only transactions simply returns $C_k$.
\subsubsection{Proof of opacity}
\label{sec:op2}
Using the same notation as in proof of opacity for Algorithm~\ref{alg:try} in Section~\ref{sec:op1}, let $E'$ denote an execution of the TM implemented by Algorithm~\ref{alg:cas} and $H'$, a linearization of the execution history $E'|_{TM}$. We construct $H'$ by assigning \emph{linearization points} to tm-operations performed in completion of $E'|_{TM}$.

The linearization point of a tm-operation $op_k$ performed by transaction $T_k$ in a completion of $E'|_{TM}$ is associated with access of a base object or a tm-event performed during the lifetime of the tm-operation as follows. 
\begin{itemize}
\item For every tm-read $op_k$ that returns a non-A$_k$ value, $\ell_{op_k}$ is chosen as the event in Line~\ref{line:casread} of Algorithm~\ref{alg:cas}, else, $\ell_{op_k}$ is chosen as invocation event of $op_k$
\item For every tm-write $op_k$ that returns, $\ell_{op_k}$ is chosen as the invocation event of $op_k$
\item For every $op_k=tryC_k$ that returns $C_k$ such that $\Wset(T_k) \neq \emptyset$, $\ell_{op_k}$ is associated with the successful acquisition of the lock on $\Wset(T_k)$ (Line~\ref{line:cas}), else if $op_k$ returns $A_k$, $\ell_{op_k}$ is associated with the invocation event of $op_k$
\item For every $op_k=tryC_k$ that returns $C_k$ such that $\Wset(T_k) = \emptyset$, $\ell_{op_k}$ is associated with Line~\ref{line:rcas}
\end{itemize}
The t-sequential history $S'$ is constructed in same manner as described in Section~\ref{sec:op1} from the above assignment for linearization points. Note that the Lemmas proven for Algorithm~\ref{alg:try} are clearly also valid for Algorithm~\ref{alg:cas}.
%\begin{lemma}
%\label{lm:total1}
%If $T_i \prec_{H'}^{DU} T_j$, then $T_i <_{S'} T_j$
%\end{lemma}
%\begin{lemma}
%\label{lm:seq1}
%If $T_i \prec_{H'}T_j$, then $T_i <_{S'} T_j$
%\end{lemma}
%\begin{lemma}
%\label{lm:readfrom1}
%$S'$ is legal
%\end{lemma}
%\\
\begin{theorem}
\label{th:kcas}
There exists a progressive opaque STM
implementation that employs
exactly one AWAR per transaction.
Moreover, no AWARs are performed in read-only or aborted transactions. 
\end{theorem}
\begin{proofsketch}
Clearly, Algorithm~\ref{alg:cas} implements an opaque STM.

Algorithm~\ref{alg:cas} is progressive since every transaction forcefully aborts either due to read-invalidation or because \emph{mCAS} returns false implying that there exists a conflicting t-object contended by a concurrent transaction. Also note that, if several transactions concurrently conflict on a single t-object, the first transaction to execute the \emph{mCAS} in Line~\ref{line:cas} is returned \emph{true} and commits. Thus, the implementation guarantees that in any set of concurrent conflicting transactions, 
at least one of the transactions commits which actually provides a stronger progress guarantee than
progressiveness or even strong progressiveness.
Indeed, a transaction $T_k$ can abort only if a concurrent \emph{committed} transaction 
modifies the value of $v_j$ for some $X_j\in\Dset(T_k)$. 
%Algorithm~\ref{alg:cas} uses \emph{invisible reads} since the tm-read operation does not perform any base-object writes.

Algorithm~\ref{alg:cas} performs a single \emph{mCAS} operation on $\Dset(T_k)$ of a transaction $T_k$ that commits during \emph{tryC$_k$}; if $T_k$ aborts, the \emph{mCAS} only performs reads of base objects. For read-only transactions, the transaction simply returns $C_k$ incurring no AWAR.
\end{proofsketch}
\subsection{Starvation-free multi-trylock}
\label{app:smul}
In this section, we define a multi-trylock object analogous to the one defined in Section~\ref{app:ml}, but whose operations are \emph{starvation-free}.
The algorithm is inspired by the \emph{Black-White Bakery Algorithm}~\cite{Gadi-Bakery} and uses a finite number of bounded registers.

The algorithm uses the following shared variables: registers $r_{ij}$ for each process $p_i$ and object $X_j$, a shared bit $\textit{color} \in \{B,W\}$, registers $LA_i \in \{0, \ldots , N\}$ for each $p_i$ that denote a \emph{Label} and $MC_i \in \{B,W\}$ for each $p_i$. 
%The priority between processes is defined as follows: If two processes have the same color, the Label with the smaller value gets the higher priority, else if two Labels have different colors, the Label whose color is different from the value of the shared bit \emph{color} gets the higher priority.

We say $(LA_i,i) < (LA_k,k)$ \emph{iff} $LA_i < LA_k$ or $LA_i=LA_k$ and $i<k$.

\begin{algorithm}[ht]
\caption{Starvation-free multi-trylock invoked by process $p_i$}\label{alg:mul2}
%\caption{Multi-trylock implementation using one multi-RAW}\label{alg:try}
  \begin{algorithmic}[1]
  	%\begin{multicols}{0}
  	{\footnotesize
	\Part{Shared variables}{
		\State $LA_i$, for each process $p_i$, initially $0$
		\State \textit{$MC_i \in \{B,W\}$} for each process $p_i$, initially $W$
%		\State $x_i\forall i=1 \ldots M$~t-objects
		%\State $v_j$, for each t-object $X_j$
%   	\State $L[1, \ldots M]$~Multi Try-lock
		%\State $Q$-set of t-objects
		\State $\textit{color} \in \{B,W\}$, initally $W$
	  	\State $r_{ij}$, for each process $p_i$ and each t-object $X_j$, initially $0$
		%\State $L$, a multi-trylock object
		%\State $L$, multi-trylock
	}\EndPart	
	\Statex
	\Part{\textit{acquire}($Q$)}{
		%\State $tr_i:=max(LA_0,\ldots LA_{n-1})+1$  \label{line:lread}
   		\ForAll{$X_j \in Q$}	
			\State \Write$(r_{ij},1)$ \label{line:lwrite}
		\EndFor
		\State $c_i:=color$ \label{line:cread}
		\State \Write$(MC_i,c_i)$ \label{line:color}
		%\If{$\exists X_j \in Rset \cup Wset,\; t\neq i$: $\{r_{tj},r_{ij}\}\cap\{2\}\neq\emptyset$ 
		%		and $r_{tj}\neq 0$}
		%	\State \Write$(r_{ij},0)$
		%	\State return $\false$
		\State \Write($LA_i, 1+max(\{LA_k)|MC_k=MC_i$\}) \label{line:label1}
		\While{$\exists j:\exists k \neq i$: $isContended(X_j)$ $\&\&$ (($LA_k \neq 0$;~($MC_k=MC_i$);~$(LA_k,k) < (LA_i,i)$) $||$ \\~~~~($LA_k \neq 0;~(MC_k \neq MC_i$);~$MC_i=color$)) } \label{line:label2}
			\State no op
		\EndWhile
		\State return $\true$ 
	}\EndPart		
	 \Statex
	 \Part{\textit{release}($Q$)}{
  		\ForAll{$X_j \in Q$} \label{line:rwrite}	
 			\State \Write$(r_{ij},0)$
		\EndFor
		%\State $Rset_k:=Wset_k:=\emptyset$
		%\State \Write($LA_i,0$)
		\If{$MC_i=B$} \label{line:color2}
			\State $\Write(color,W)$
		\Else
			\State $\Write(color,B)$
		\EndIf
		\State \Write($LA_i,0$)
		\State return \ok
	}\EndPart
	\Statex
	\Part{{\it \textit{isContended}($X_j$)}}{
		\If{$\exists p_t: r_{tj} \neq 0, t \neq i$} \label{line:isl}
			\State return $\true$
		\EndIf
		\State return $\false$
	}\EndPart			
	}
%\end{multicols}
  \end{algorithmic}
\end{algorithm}

A starvation-free multi-trylock implementation satisfies the following properties:
\begin{itemize}
\item \emph{Mutual-exclusion}: For any object $X_j$, and any execution
  $\pi$, there exists at most one process that \emph{holds} a lock on
  $X_j$ after $\pi$. 

\item \emph{Progress}:  Let $\pi$ be any execution that contains
  $\textit{acquire}(Q)$ by process $p_i$. If no other process $p_k, k\neq i$ contends infinitely long on some 
$X_j \in Q$, then $\textit{acquire}(Q)$ returns \emph{true} in $\pi$. 

\item Let $\pi$ be any execution that contains $\textit{isContended}(X_j)$ invoked by $p_i$.

\begin{itemize}
\item If $X_j$ is locked by $p_{\ell};\ell \neq i$ during the complete execution of
$\textit{isContended}(X_j)$ in $\pi$, then
$\textit{isContended}(X_j)$ returns $\true$.

\item If $\forall \ell \neq i$, $X_j$ is never contended by $p_{\ell}$ during the execution of
$\textit{isContended}(X_j)$ in $\pi$, then
$\textit{isContended}(X_j)$ returns $\false$.
\end{itemize}
\end{itemize}
\begin{lemma}
\label{lm:mutex}
In every execution $\pi$ of Algorithm~\ref{alg:mul2}, if $p_i$ \emph{holds} a lock on some object $X_j$ after $\pi$, then one of the following conditions must hold:
\begin{enumerate}
\item[(1)]
for some $k\neq i$; $LA_k \neq 0$, if $MC_k=MC_i$, then $(LA_k,k) > (LA_i,i)$
\item[(2)]
for some $k\neq i$; $LA_k \neq 0$, if $MC_k \neq MC_i$, then $MC_i \neq color$
\end{enumerate}
\end{lemma}
\begin{proof}
In order to hold the lock on $X_j$, some process $p_i$ writes $1$ to $r_{ij}$, writes a value, say $W$ to $MC_i$ and reads the Labels
of other processes that have obtained the same color as itself and generates a Label greater by one than the maximum Label read (Line~\ref{line:label1}). 
Observe that until the value of the \emph{color} bit is changed, all processes read the same value $W$. The first process $p_i$
to hold the lock on $X_j$ changes the \emph{color} bit to $B$ when releasing the lock and hence the value read by all subsequent processes will be $B$ until it is changed again. Now consider two cases:  
%Now consider two cases: two processes obtain the same \emph{color} in which case their \emph{Labels} are comparable (or) they obtain different \emph{colors} (in which case, the \emph{Labels} are incomparable).
\begin{enumerate}
\item[(1)]
Assume that there exists a process $p_k$, $k \neq i$, $LA_k \neq 0$ and $MC_k=MC_i$ such that $(LA_k,k) < (LA_i,i)$, but $p_i$
holds a lock on $X_j$ after $\pi$. Thus, $isContended(X_j)$ returns \emph{true} to $p_i$ because $p_k$ writes to $r_{kj}$ (Line~\ref{line:lwrite}) before writing to $LA_k$ (Line~\ref{line:label1}). 
By assumption, $(LA_k,k) < (LA_i,i);LA_k>0$ and $MC_i=MC_k$, but the conditional in Line~\ref{line:label2} returned \emph{true} to $p_i$ 
without waiting for $p_k$ to stop contending on $X_j$---contradiction. %Similarly, if ($LA_k < LA_i$), we can see the contradiction again.
\item[(2)]
Assume that there exists a process $p_k$, $k \neq i$, $LA_k \neq 0$ and $MC_k \neq MC_i$ such that $MC_i=color$, but $p_i$ holds a lock on $X_j$ after $\pi$. Again, since $LA_k >0$, $isContended(X_j)$ returns \emph{true} to $p_i$, $MC_k \neq MC_i$ and $MC_i=color$, but the conditional in Line~\ref{line:label2} returned \emph{true} to $p_i$ without waiting for $p_k$ to stop contending on $X_j$---contradiction.
\end{enumerate}
\end{proof}
\begin{theorem}
\label{th:lock2}
Algorithm~\ref{alg:mul2} is an implementation of multi-trylock object
in which every operation is starvation-free and incurs at most four RAWs.
\end{theorem}
\begin{proof}
Denote by $L$ the shared object implemented by
Algorithm~\ref{alg:mul2}.

Assume, by contradiction, that $L$ does not provide mutual-exclusion:
there exists an execution $\pi$ after which processes $p_i$ and $p_k$, $k\neq i$
\emph{hold} a lock on the same object, say $X_j$. Since both $p_i$ and $p_k$ have performed the write to $LA_i$ and $LA_k$ resp. in Line~\ref{line:label1}, $LA_i, LA_k >0$. Consider two cases:
\begin{enumerate}
\item[(1)]
If $MC_k=MC_i$, then from Condition $1$ of Lemma~\ref{lm:mutex}, we have $(LA_k,k) < (LA_i,i)$ and $(LA_k,k) > (LA_i,i)$---contradiction.
\item[(2)]
If $MC_k \neq MC_i$, then from Condition $2$ of Lemma~\ref{lm:mutex}, we have $MC_i \neq color$ and $MC_k \neq color$ which
implies $MC_k=MC_i$---contradiction.
\end{enumerate}

$L$ also ensures progress. If process $p_i$ wants to hold the lock on an object $X_j$ i.e. invokes $\textit{acquire}(Q), X_j \in Q$, it checks if any other process $p_k$ holds the lock on $X_j$. If such a process $p_k$ exists and $MC_k=MC_i$, then clearly $isContended(X_j)$ returns \emph{true} for $p_i$ and $(LA_k,k) < (LA_i,i)$. Thus, $p_i$ fails the conditional in Line~\ref{line:label2} and waits until $p_k$ releases the lock on $X_j$ to return \emph{true}. However, if $p_k$ contends infinitely long on  $X_j$, $p_i$ is also forced to wait indefinitely to be returned \emph{true} from the invocation of $\textit{acquire}(Q)$. The same argument works when $MC_k \neq MC_i$ since when $p_k$ stops contending on $X_j$, $isContended(X_j)$ eventually returns \emph{false} for $p_i$ if $p_k$ does not contend infinitely long on $X_j$.

All operations performed by $L$ are starvation-free. Each process $p_i$ that successfully holds the lock on
an object $X_j$ in an execution $\pi$ invokes $\textit{acquire}(Q), X_j \in Q$, obtains a color and 
chooses a value for $LA_i$ since there is no way to be blocked while writing to $LA_i$. 
The response of operation $\textit{acquire}(Q)$ by $p_i$ is only delayed if there exists a concurrent invocation of $\textit{acquire}(Q'),X_j \in Q'$ by $p_k$ in $\pi$. In that case, process $p_i$ waits until $p_k$ invokes $\textit{release}(Q)$ and writes $0$ to $r_{kj}$ and eventually holds the lock on $X_j$. The implementation of \emph{release} and \emph{isContended} are wait-free operations (and hence starvation-free) since they contains no unbounded loops or waiting statements.

The implementation of $\textit{isContended}(X_j)$ only reads base
objects. 
The implementation of $\textit{release}(Q)$ writes to a series of base objects (Line~\ref{line:rwrite}) and then reads a base object (Line~\ref{line:color2}) incurring a single RAW. 
The implementation of $\textit{acquire}(Q)$ writes to base objects (Line~\ref{line:lwrite}), reads the shared bit $color$ (Line~\ref{line:cread})---one RAW, writes to a base object (Line~\ref{line:color}), reads the Labels (Line~\ref{line:label1})---one RAW, writes to its own Label and finally performs a sequence of reads when evaluating the conditional in Line~\ref{line:label2}---one RAW. 

Thus, Algorithm~\ref{alg:mul2} incurs at most four RAWs.
\end{proof}
\subsection{Strong progressive implementation with constant RAWs}
\label{app:sraw}
Let $CObj_H(T_i)$ denote the set of t-objects over which transaction $T_i \in \parts(H)$ conflicts with any other transaction in history $H$
i.e. $X \in CObj_H(T_i)$, if there exists a transaction $T_k \in \parts(H)$, $k\neq i$, such that $T_i$ conflicts with $T_k$ on $X$ in $H$. 
Then, $CObj_H(Q)=\{CObj_H(T_i) |\forall T_i \in Q\}$, denotes the union of sets $CObj_H(T_i)$ for all transactions in $Q$.

Let $CTrans(H)$ denote the set of non-empty subsets of $\parts(H)$ such that a set $Q$ is in $CTrans(H)$ if no transaction in $Q$ conflicts with a transaction not in $Q$.
\begin{definition}
\label{def:sprog}
A TM implementation $M$ is \emph{strongly progressive} if $M$ is weakly progressive 
and for any history $H$ of $M$, there does not exist a prefix $H'$ of $H$ in which every set $Q \in CTrans(H')$ 
of transactions that are live in $H'$ such that $|CObj_{H'}(Q)| \leq 1$, every transaction in $Q$ is forcefully aborted in $H$.
\end{definition}
Algorithm~\ref{alg:gp} describes the implementation of the \emph{tryC} operation of a strongly progressive, opaque TM. The only modification over
the \emph{tryC} implementation of Algorithm~\ref{alg:try} is that in Algorithm~\ref{alg:gp}, every transaction with $|\Rset|=\emptyset$ eventually commits. The \emph{read}, \emph{write}, \emph{tryA} and \emph{isAbortable} operations are the same as in Algorithm~\ref{alg:try}. 
\begin{algorithm}[t]
\caption{Strongly progressive, opaque STM: the implementation of $T_k$ executed by $p_i$}\label{alg:gp}
\begin{algorithmic}[1]
  	\begin{multicols}{1}
  	{\footnotesize
	\Part{Shared variables}{
		%\State $LA_i$, for each process $p_i$
%		\State $x_i\forall i=1 \ldots M$~t-objects
		\State $v_j$, for each t-object $X_j$
%   	\State $L[1, \ldots M]$~Multi Try-lock
		%\State $Q$-set of t-objects
	  	%\State $r_{ij}$, for each process $p_i$ and each t-object $X_j$
		\State $L$, a starvation-free multi-trylock object
		%\State $L$, multi-trylock
	}\EndPart	
	\Statex
	\Part{\TryC$_k$()}{
		\If{$|\Wset(T_k)|= \emptyset$}
			\State return $C_k$ \label{line:return}
		\EndIf
		\State locked $:= L.\textit{acquire}(\Wset(T_k))$\label{line:acq} 
		%\If{not locked} \label{line:abort2} 
	 	%	\State return $A_k$
	 	%\EndIf
		\If{isAbortable()} \label{line:abort3}
			\State $L.\textit{release}(\Wset(T_k))$ 
			\State return $A_k$
		\EndIf
		\ForAll{$X_j \in \Wset(T_k)$}
	 		 \State $\Write(v_j,(\textit{nv}_j,k))$ \label{line:write}
	 	\EndFor		
		\State $L.\textit{release}(\Wset(T_k))$   		
   		\State return $C_k$
   	 }\EndPart		
	 
	}\end{multicols}
  \end{algorithmic}
\end{algorithm}
\begin{theorem}
\label{tm:gpprog}
Algorithm~\ref{alg:gp} implements a strongly progressive TM
\end{theorem}
\begin{proof}
Every transaction $T_k$ in a TM $M$ whose tm-operations are defined by Algorithm~\ref{alg:gp} can be aborted in the following scenarios
\begin{itemize}
\item Read-validation failed in \textit{$read_k$} or \textit{tryC$_k$}
\item \textit{$read_k$} or \textit{tryC$_k$} returned $A_k$ because $X_j \in \Rset(T_k)$ is locked (belongs to write set of a concurrent transaction)  
%\item \textit{L.acquire($\Wset(T_k)$)} returned \textit{false} in Line~\ref{line:abort2} of Algorithm~\ref{alg:gp}
\end{itemize}
Thus, Algorithm~\ref{alg:gp} implements a weakly progressive TM (From Lemma~\ref{lm:progtm}).

To show Algorithm~\ref{alg:gp} also implements a strongly progressive STM, we need to show that for every set of transactions that concurrently contend on a single t-object, at least one of the transactions is not aborted. 

Consider transactions $T_i$ and $T_k$ that concurrently attempt to execute $tryC_i$ and $tryC_k$ such that $X_j \in \Wset_i \cup \Wset_k$. Consequently, they both invoke the \emph{acquire} operation of the multi-trylock  (Line~\ref{line:acq}) and thus, from Theorem~\ref{th:lock2}, both $T_i$ and $T_k$ must commit eventually.
Also, if validation of a tm-read in $T_k$ fails, it means that the t-object is overwritten by some transaction $T_i$ such that $T_i$ precedes $T_k$, implying at least one of the transactions commit. Otherwise, if some t-object $X_j \in \Rset(T_k)$ is locked and returns \emph{abort} since the t-object is in the write set of a concurrent transaction $T_i$. While it may still be possible that $T_i$ returns $A_i$ after acquiring the lock on $\Wset_i$, strong progressiveness only guarantees progress for transactions that conflict on at most one t-object. Thus, in either case, for every set of transactions that conflict on at most one t-object, at least one transaction is not forcefully aborted.
\end{proof}
\begin{reptheorem}[Theorem~\ref{th:sprogop}]
%\label{th:sprogop}
There exists a strongly progressive single-version opaque STM
implementation with starvation-free operations that uses invisible reads and employs
at most four RAWs per transaction.
Moreover, no RAWs are performed in read-only transactions. 
\end{reptheorem}\\
\\
\begin{proof}
The correctness of Algorithm~\ref{alg:gp} clearly follows from the proof of opacity presented in Section~\ref{sec:op1} for Algorithm~\ref{alg:try}.
From Theorem~\ref{tm:gpprog}, it is also strongly progressive.
%Algorithm~\ref{alg:gp} uses \emph{invisible reads} since the tm-read operation does not perform any base-object writes.

Any process executing a transaction $T_k$ holds the lock on $\Wset(T_k)$ only once during \emph{tryC$_k$}. If
$|\Wset(T_k)|=\emptyset$, then the transaction simply returns $C_k$
incurring no RAW's. Thus, from Theorem~\ref{th:lock2}, Algorithm~\ref{alg:gp} incurs at most four RAWs per updating transaction and no RAW's are performed in read-only transactions.
\end{proof}
\section{RAW/AWAR cost of probabilistically permissive STMs}
\label{app:prob}
\begin{theorem}
\label{th:prob}
Let $M$ be a probabilistically permissive opaque STM implementation.
Then, for any $m$, there exists with positive probability, an execution in which a read-only transaction $T_i$ contains   
$\Omega(m)$ non-overlapping RAWs or AWARs on base objects where $m=|\Rset(T_i)|$.
\end{theorem}
\begin{proof}
For the proof, note that we only need to show that there exists an execution of the probabilistically permissive TM that is the same
as the execution of a permissive TM, Then, the construction and arguments used in the proof of Theorem~\ref{th:perm} can be extended for the probabilistic case. 

Let $E$ denote the execution depicted in Figure~\ref{fig:H} where $T_3$ performs a read of $X_1$, then $T_2$ performs a write on $X_1$ and commits, and finally $T_1$ performs a series of reads on $X_1, \ldots , X_m$. We proceed by induction by considering $R_1(X_k)$, 
the $k$-th read of $T_1$, $2\leq k\leq m$. 
\begin{enumerate}
\item[(1)]
Imagine an extension of $E$, denoted by  $E'$, in which $T_3$ performs a $W_3(X_k)$ immediately after $R_1(X_k)$
and then tries to commit.
A serialization of $H'=E'|_{TM}$ should obey
$T_3\prec_{H'}^{DU} T_2$ and $T_2\prec_{H'} T_1$.
The execution of $R_1(X_k)$ does not modify base objects, hence, $T_3$
does not observe $R_1(X_k)$ in $E'$.
In a probabilistically permissive TM, the tm-operation $W_3(T_k)$ can return one of the following values $A_k$ or $ok_k$.
Note that this response is chosen by sampling uniformly at random from the set of possible return values, thus, there exists a positive probability that
$T_3$ commits successfully (when it returns $ok_k$).
But since $T_1$ performs $R_1(X_k)$ before $T_3$ commits and $T_3$ updates $X_k$, we also have $T_1\prec_{H'}^{DU} T_3$.
Thus, $T_3$ cannot precede $T_1$ in any serialization and we establish a contradiction.
Consequently, there exists with positive probability, an execution in which each $R_1(X_k)$, $2 \leq k \leq m$ performs a write to a base-object.
\item[(2)]
Let $\pi$ be a fragment of $E$ that represents the complete execution of $R_1(X_k)$. 
Clearly, there exists with positive probability, an execution in which $\pi$ contains a write to a base-object. Let $\pi_j$ be 
the first write to a base-object in $\pi$ and $\pi_w$, the shortest fragment of $\pi$ 
that contains the atomic section to which $\pi_j$ belongs, else if $\pi_j$ 
is not part of an atomic section, $\pi_w=\pi_j$. Thus, $\pi$ can be represented as $\pi_s\cdot \pi_w\cdot \pi_f$.
 
Suppose that $\pi$ does not contain a RAW or AWAR. 
Since $\pi_w$ does not contain an AWAR (atomic write-after-read), there are no read events 
in $\pi_w$ that precede $\pi_j$. Thus, $\pi_j$ is the first base-object event in $\pi_w$. 
Consider the execution fragment $\pi_s\cdot \rho$, where $\rho$ is the complete execution 
of $\{W_3(X_k)$, $TC_3 \}$ by transaction $T_3$. 
By Definition~\ref{def:perm}, such an execution exists with positive probability in which $T_3$ commits. Since $\pi_s$ does not perform any base-object write, $\pi_s \cdot\rho$ is indistinguishable to $T_3$ from $\rho$. 

Also, by our assumption, $\pi_w\cdot \pi_f$ contains no RAW i.e. 
any read performed in $\pi_w\cdot \pi_f$ 
can only be applied to base objects previously written in $\pi_w\cdot \pi_f$. 
Thus, in a probabilistically permissive TM in which responses to tm-operations are chosen by independent coin-tosses, 
there exists with positive probability, an execution $\pi_s \cdot\rho \cdot \pi_w\cdot \pi_f$ 
that is indistinguishable to $T_1$ from $\pi$.
However, in $\pi_s \cdot\rho \cdot \pi_w\cdot \pi_f$, $T_3$ commits (as in $\rho$) but 
$T_1$ ignores the value written by $T_3$ to $X_k$.  
But $T_3$ can only be serialized before $T_1$---contradiction.
\end{enumerate}
\end{proof}

\end{document}
%%% Local Variables:
%%% mode: latex
%%% mode: flyspell
%%% Local IspellDict: "american"
%%% mode: outline-minor
%%% End: